%% file: AAA_main_Arxive.tex
\documentclass{article}
\usepackage{arxiv}

\usepackage[utf8]{inputenc}
\usepackage[T1]{fontenc} 
\usepackage{setspace}
\usepackage{amsmath}
\usepackage{xcolor}
\usepackage{natbib}
\usepackage{graphicx}
\usepackage{bbm}
\usepackage{enumitem}
\usepackage{soul}
\usepackage{subcaption}
\usepackage{enumitem}
\usepackage{tikz}
\usetikzlibrary{positioning}
\usepackage{caption}  
\usepackage{algpseudocode}
\usepackage{adjustbox}
\usepackage{mathtools}
\usepackage{algorithm}
\usepackage{graphicx}
\usepackage{caption}
\usepackage{adjustbox}

\newcommand{\yb}{\textbf{y}}
\newcommand{\ybt}{\textbf{y}_{1:T}}
\newcommand{\ybo}{\textbf{y}_{o}}
\newcommand{\ybm}{\textbf{y}_{m}}

\newcommand{\zbt}{\textbf{z}_{0:T}}

\newcommand{\mbt}{\textbf{m}_{1:T}}
\newcommand{\bb}{\boldsymbol{\beta}}
\newcommand{\zz}{\boldsymbol{\zeta}}
\newcommand{\xx}{\boldsymbol{\xi}}
\newcommand{\tb}{\boldsymbol{\theta}}
\newcommand{\pp}{\boldsymbol{\phi}}

\makeatletter
\newcommand*{\centernot}{%
  \mathpalette\@centernot
}
\def\@centernot#1#2{%
  \mathrel{%
    \rlap{%
      \settowidth\dimen@{$\m@th#1{#2}$}%
      \kern.5\dimen@
      \settowidth\dimen@{$\m@th#1=$}%
      \kern-.5\dimen@
      $\m@th#1\not$%
    }%
    {#2}%
  }%
}

\usepackage{amsthm}

\newtheorem{proposition}{Proposition} 
\newtheorem{definition}{Definition} 
\title{Bayesian Nonhomogeneous hidden Markov models to leverage routine in physical activity monitoring with informative wear time}

\author
{Beatrice Cantoni  \\
Department of Statistics and Data Sciences, University of Texas at Austin \\ \texttt{beatrice.cantoni@utexas.edu}
\AND
Savannah V. Rauschendorfer  \\
Department of Health, Human Performance, and Recreation, Baylor University \\ \texttt{savannah\_rauschendor@baylor.edu}
\AND 
Michael E. Roth
\\
Department of Pediatrics, The University of Texas MD Anderson Cancer Center \\ \texttt{mroth1@mdanderson.org} 
\AND
J. Andrew Livingston
\\
Department of Sarcoma Medical Oncology, The University of Texas MD Anderson Cancer Center \\ \texttt{jalivingston@mdanderson.org} 
\AND
Eugenie S. Kleinerman
 \\
Department of Pediatrics Research, The University of Texas MD Anderson Cancer Center \\ \texttt{ekleiner@mdanderson.org}
\AND
Corwin M. Zigler\\
Department of Biostatistics, Brown University School of Public Health \\ \texttt{corwin\_zigler@brown.edu } }

\begin{document}

\maketitle
\newpage 
\begin{abstract}
Missing data is among the most prominent challenges in the analysis of physical activity (PA) data collected from wearable devices, with the threat of nonignorabile missingness arising when patterns of device wear relate to underlying activity patterns.  We offer a rigorous consideration of assumptions about missing data mechanisms in the context of the common modeling paradigm of state space models with a finite, meaningful, set of underlying PA states. Focusing in particular on hidden Markov models, we identify inherent limitations in the presence of missing data when covariates are required to satisfy common missing data assumptions.   In response to this limitation, we propose a Bayesian non-homogeneous state space model that can accommodate covariate dependence in the transitions between latent activity states, which in this case relates to whether patients' routine behavior can inform how they transition between PA states and thus support imputation of missing PA data. We show the benefits of the proposed model for missing data imputation and inference for relevant PA summaries.  Our development advances analytic capacity to confront the ubiquitous challenge of missing data when analyzing PA studies using wearables. We illustrate with the analysis of a cohort of adolescent and young adult (AYA) cancer patients who wore commercial Fitbit devices for varying durations during the course of treatment.
\end{abstract}

\keywords{ Data missingness \and Physical activity tracking \and Bayesian state-space models}

\input{main_text}

\bibliographystyle{apalike}
\bibliography{mybiblio}

\newpage

\appendix

\input{supplementary_material}

\end{document}

%% file: main_text.tex
\section{Introduction}
\label{s:intro}
Recent technological developments have made physical activity (PA) tracking devices broadly accessible, with increased adoption in studies where improved understanding of PA may inform the design of clinical interventions \citep{st2021use, strain2022considerations}. 
Analyzing PA data from wearable devices confronts several well know challenges, such as device reliability, data dimension, compatibility of consumer devices and, as is the focus of this work, data missingness \citep{balbim2021, hicks2019best, migueles2017accelerometer, karas2019accelerometry}.  
While missing data with PA wearables is ubiquitous, current literature often lacks formal consideration of missing data assumptions, with frequent and often implicit reliance on data being Missing at Random (MAR) or even Missing Completely at Random (MCAR) \citep{alhaddad2023longitudinal}. 
In reality, missing data mechanisms are likely related to underlying PA if, for example, a person's decision of whether to actually wear their activity tracking device is related to their intention to engage in certain types of activity. Formalizing assumptions about missing data mechanisms could guide PA modeling efforts that incorporate covariates relevant to the missing data and set the direction for strategies for missing data imputation.

We consider a framework in which raw sensor data have already been preprocessed into interpretable metrics such as steps and heart rate per unit time, as delivered by the commercial wearable devices in our motivating application.
Many models have been proposed for PA with this type of data, a very common choice being state space models with discrete states, most commonly in the form of a Hidden Markov Model (HMM) designed for activity type categorization and clustering.
Such models assume that the observed PA is the realization of an underlying classification of activity, such as low, medium, or high intensity activity, with model structure placed on how a person transitions between activity states \citep{langrock2013combining, de2020mixture, huang2018hidden}.
We formalize assumptions for missing data in the context discrete-space HMMs as well as general state space models, offering a framework to asses the role of routinely-available covariates in formalizing missing data assumptions.
In so doing, we isolate certain limits of HMMs in their ability to accommodate realistic missing data assumptions in PA studies.
In particular, we show how typical HMMs lack the ability to support the crucially important ignorability condition underlying the MAR assumption whenever missingness is expected to relate to underlying activity category. 
This has not, to our knowledge, been previously recognized or formulated in the HMM or PA literature. 

We hence extend classical HMMs to incorporate additional information to support realistic assumptions about ignorability of the missing data mechanism. 
Our framework can accommodate any covariates that relate to latent behavior and missingness, but we focus on indicators of time-of-day under the presumption that patients' routine activity patterns might dictate device wear and PA.
We leverage recent developments in the literature for Bayesian Nonhomogeneous HMMs (NHMMs) \citep{holsclaw2017bayesian} to include covariate information in the transition model through Pólya-Gamma data augmentation and adapt this model for PA tracking. We highlight and address practical implementation challenges that have not received pointed focus in previous work on Bayesian NHMMs, and relate these issues to the more general problem of Markov chain Monte Carlo with imbalanced categorical data \citep{johndrow2019mcmc}. 
Through simulations, we show how our proposed methodology outperforms state space models currently used in PA modeling with improved parameter estimation and missing data imputation when data can only be assumed MAR conditional on observed covariates impacting transitions among activity states.

We deploy the proposed methods in an analysis of PA among adolescent and young adult (AYA) cancer patients who used wrist worn devices manufactured by Fitbit Inc. in free living conditions to record  minute-level heart rate and step counts over periods ranging from weeks to months. 
Our results indicate that information coming from individual behavioural pattern and encoded in time-of-day covariates is indeed relevant, and that the proposed Bayesian NHMM procedure refines the results of an analysis - in terms of missing data imputation and clincal PA summaries - over what would be available from a more typical HMM with limited ability to leverage covariates relating to both PA and device wear. 

\section{Routine activity and wear patterns in a study of AYA cancer patients
}
\label{s:data}
PA data from wearable devices could be missing for a variety of reasons, but most notable in the context of this work is missing PA measures for times when a person does not wear their device, which invariably arises when data are recorded in free living conditions. Device non-wear could arise for a variety of reasons rangning from the simple need to charge the device to behavioral choices that are related underlying PA. For example, it is reasonable to suspect that some people are more likely to wear their device when exercising, which presents a clear threat of data being Missing Not at Random (MNAR). As we clarify in the subsequent, such circumstances might lead to biased inference and inaccurate conclusions with standard approaches to modeling PA. 

Evidence that missing data are MNAR is apparent in the analysis of AYA cancer patients, as evidenced by observing how patients' missingness and PA outcomes vary throughout the time of day. A visualization is given in Figure \ref{fig:heatmaps}, which depicts patterns of device wear and step count for four AYA cancer study participants.
Figure \ref{fig:nonwear} illustrates how the apparent choice to wear the device is related to the hour of the day, showing heterogeneity across patients.
Figure \ref{fig:steps_heatmap} shows how the hour of the day also influences step counts, a fact that has been shown to be relevant in previous works \citep{ren2022measuring, ren2023combining}.
Together, these figures underscore both the heterogeneity of patterns across individuals and the potential threat to common missing data assumptions generated by the interplay between device wear behavior and activity patterns.  
\begin{figure}[htbp]
    \centering
    \begin{subfigure}[b]{0.49\textwidth} 
        \includegraphics[width=\textwidth]{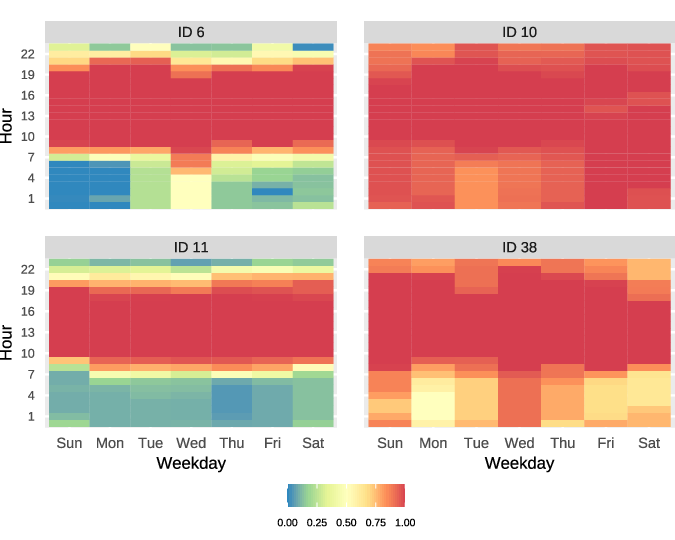}
        \caption{Wear time ratio.}
        \label{fig:nonwear}
    \end{subfigure}
    \hfill 
    \begin{subfigure}[b]{0.49\textwidth} 
        \includegraphics[width=\textwidth]{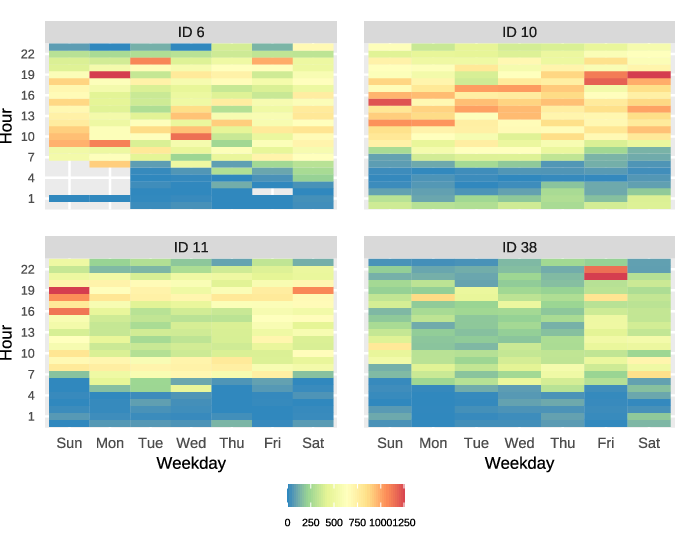}
        \caption{Steps.}
        \label{fig:steps_heatmap}
    \end{subfigure}
    \caption{Summary of individual daily routine patterns for four individuals: (a) proportion of time interval the device was worn across day of week (b) the average number of steps during each time interval of each day of week. Grey areas correspond to hours with no data during the study period.}
    \label{fig:heatmaps}
\end{figure}

\section{HMMs for modeling PA
}
\label{s:HMM}
Denote by $T$ the total number of observations considered in the dataset for a generic patient and by $y_t$ the vector of PA observations at time $t$, for $t=1, \dots, T$. 
Denote by $\ybt$ the matrix containing all such PA outcomes.
For example, in our AYA cancer study, $\ybt$ would represent a $T \times 2$ matrix where the $t^{th}$ row contains the heart rate and step count for each time interval. 
Finally, denote by $x_{1:T}$ the matrix of available exogenous covariates for this generic patient, and by $x_t$ a vector of covariates at a single time, $t$. 

State space models are a class of time series models characterized in the first place by the assumption that the observed outcome $y_t$ depends on an underlying and unknown time-indexed \textit{state}, which here we denote as $z_t$, with $\zbt$ the sequence of those underlying states at times $1:T$ plus an initial state, $z_0$. Each state $z_t$ for $t>0$ is associated with state-specific parameters that dictate the observed $y_t$. The dependence between the underlying state and the observed outcome is named \textit{emission distribution}, and typically $y_t$ does not depend on $y_{1:t-1}$ conditional on $z_t$. Relationships between states at different points in time are governed by a \textit{transition distribution}. A state space model can be expressed as:
\begin{equation}
\label{eq:lik}
    p(\ybt, \zbt) = p_{\pi}(z_0)   \prod_{t=1}^{T} p_q (z_{t} \mid \textbf{z}_{0: t-1}) p_{\psi}(y_t \mid z_t)
\end{equation}
\begin{equation}
    p(\ybt) = \int p_{\pi}(z_0)   \prod_{t=1}^{T} p_q (z_{t} \mid \textbf{z}_{0: t-1})p_{\psi}(y_t \mid z_t)  d \zbt
\end{equation}
where $\psi$ are the parameters governing the emission distribution, $\pi$ governs the probability of an initial state, $z_0$, and $q$ are the parameters governing the transition distribution. Different specifications of emission and transition distributions correspond to different state space models that invite different estimation procedures, some developed within a Bayesian framework. 
For example, a Gaussian dynamic linear model is characterized by Gaussian $p_q(\cdot)$ and $p_\psi(\cdot)$ over a continuous state space, supporting straightforward posterior estimation strategy via the Kalman Filter. 
PA tracking data are more commonly modeled via discrete state space models having a finite set of states, each corresponding to a categorization of activity ranging from sedentary to intense activity.
Within this subclass of state space models, Hidden Markov Models (HMMs) are the most common \citep{langrock2013combining, witowski2014using,de2020mixture}.
In an HMM, the states $z_t$ evolve with a Markov property that at each time $t$, the current state depends on the previous states $\textbf{z}_{0:t-1}$ through $z_{t-1}$ only, which can correspond to the inherently smooth nature of PA variation across levels of activity. With this property, the transition distribution of an HMM takes the form:
\begin{equation}
    p_q(z_t = j \mid z_{t-1} = i, \textbf{z}_{0:t-1}) = p_q(z_t = j \mid \boldsymbol{z}_{t-1} = i) = q_{ij} \label{eq:HMM2}
\end{equation}
with $q_{ij}$ the entries of a transition matrix $Q$ corresponding to the probability of transitioning from $z_{t-1} = i$ to $z_{t} = j$.  
Note that covariates do not play a role in the most common expression of HMMs.
While covariates could fairly easily be incorporated into the emissions distributions, including them in the Markov chain of underlying states is not straightforward \citep{zucchini2009hidden, altman2007mixed}, especially in the Bayesian context.

\section{Formalizing Missing Data Assumptions for State Space Models}
\label{s:ign}
We formalize missing data assumptions for a broader class of state space models before resuming focus on the HMMs for PA in Section \ref{s:HMM}. Formalizing missing data assumptions requires augmenting the state space model with additional consideration of the {\it missing data model} describing the mechanism that governs which measures are actually observed (in such contexts, the state space model might be referred to as the \textit{response model}).  
Augmenting the state space model with the missing data model permits explicit formulation of ignorability, which refers to a set of conditions that dictate whether the missing data model can be ignored for making inference about the state space model. Key among the conditions for ignorability is whether the data can be assumed ``missing at random'' (MAR), an assumption that, in practice, typically relies on the presence of covariates that relate to why a particular value is missing.
Detailed descriptions of ignorable missing data mechanisms appear in \cite{rubin1976inference} and \cite{daniels2008missing}, but have received little formal consideration in state space models. 
\cite{vidotto2020multiple} mention information that should intuitively be included in a missing data imputation strategy, without providing formal justification. 
\cite{speekenbrink2021ignorable} present methods that do not pursue ignorability and posit models for the missing data mechanism in an HMM.  
We focus on the case where covariates presumed to relate to missingness are available, and formalize ignorability conditions in the context of state space models for PA data collected via wearable devices. 
For this type of data, it is important to note that the missing data mechanism could derive from features of the device measurement (e.g., the device has a limit of detection) or features of the person's behavior (e.g., a person decides to wear or not wear the device).  
While we consider both, we regard the latter as the most salient threat to satisfying the ignorability conditions.  

\subsection{Ignorability conditions for state space models}
\label{s:ign2} 
Let $\tb$ denote the vector of unknown parameters indexing a state space model, i.e. $\tb = (\pi, \psi, q)$ in expression (\ref{eq:lik}).
Denote by $\pp$ the parameters governing the missing data model determining whether a particular $y_t$ is observed, with $\mbt^y$ the corresponding vector of binary entries indicating the values of $t$ for which $y_t$ is missing, where we assume for simplicity that when $y_t$ is a vector, all elements are either observed or missing (corresponding to device wear/nonwear).
The data vector $\ybt$ is partitioned into the set of missing values, $\ybm : = \{y_t; m_t^y = 1\}$, and the set of observed values,  $\ybo : = \{y_t; m_t^y = 0\}$.  Also consider $x_{1:T}$ to be a vector of covariates that may be part of either the state space model or the missing data mechanism.
The full likelihood comprised of both the state space model and the missing data model is:
\begin{equation}
    L \left(\boldsymbol{\theta}, \boldsymbol{\phi} \mid \ybo, \mbt^y,  x_{1:T} \right) \\ 
    = \int \int  p\left(\mbt^y,  \zbt, \ybo, \ybm \mid \boldsymbol{\theta}, \boldsymbol{\phi}, x_{1:T}\right) d \ybm d \zbt 
\end{equation}
The definition of ignorability in this context can be stated as follows:
\begin{definition}
\label{def}
For a state space model with outcome $\ybt$, available observations $\ybo$, missing observations $\ybm$, latent states $\zbt$, state space model parameters $\boldsymbol{\theta}$, parameters governing the missing mechanism $\pp$, and  other observables $x_{1:T}$, a missing data mechanism is said to be ignorable for the purposes of posterior inference if:
\begin{enumerate}
    \item \label{cond1} The full data parameter $(\boldsymbol{\theta}, \pp)$ can be decomposed so that $\tb$ indexes the full-data likelihood $p(\ybt \mid \boldsymbol{\theta})$ and $\pp$ indexes the missing data mechanism, $p( \mbt^y \mid \boldsymbol{\phi})$;
    \item \label{cond2} The parameters $\boldsymbol{\theta}$ and $\pp$ are independent a priori, i.e.
    $
p(\boldsymbol{\theta},\boldsymbol{\phi}) = p(\boldsymbol{\theta})p(\boldsymbol{\phi})
    $;
    \item \label{cond:3} $\mbt^y \perp (\ybm, \zbt) \mid x_{1:T}, \ybo $, which can be written as
    \begin{enumerate}[label=\arabic{enumi}\alph*)]
        \item \label{cond:a} $\mbt^y \perp \ybm \mid \zbt, \ybo, x_{1:T}$ 
    \item \label{cond:b} $\mbt^y \perp \zbt \mid x_{1:T}, \ybo$
    \end{enumerate}
\end{enumerate} 

\end{definition}
Conditions \ref{cond1} and \ref{cond2} are standard technical conditions, with the crucial condition for state space models being the independence of the missing data indicator and {\it both} $\ybm$ and the latent $\zbt$ in condition \ref{cond:3}. 
This is the analog of a standard assumption that the missing data are MAR (conidtional on $x_{1:T}$), but extended here to encode the inherent dependence on the latent variable of the state space model. 
We decompose this relationship into Conditions \ref{cond:a} and \ref{cond:b} to permit an assessment of different types of threats to ignorability assumption for state space models for PA.  

To clarify the role of the conditions in Definition \ref{def}, we start with the observed data posterior for all unknown parameters of the state space and missing data model:
\begin{align*}
&p(\tb, \pp \mid \ybo, \mbt^y,  x_{1:T}) \propto p(\tb, \pp) L \left(\boldsymbol{\theta}, \boldsymbol{\phi} \mid \ybo, \mbt^y , x_{1:T}\right). \\
\intertext{Using condition \ref{cond2} and clarifying integration over unobserved quantities:}
=& p(\tb) p(\pp) \int  \int  p\left(\mbt^y,\zbt,\ybo, \ybm \mid \boldsymbol{\theta}, \boldsymbol{\phi},  x_{1:T} \right) d \ybm d \zbt \\
\intertext{and adding condition \ref{cond1} 
}
&= p(\tb) p(\pp) \int \int p\left(\mbt^y \mid \ybo, \ybm, \zbt, \boldsymbol{\phi}, x_{1:T}\right) \\
& \hspace{1cm} \times p\left(\ybo, \ybm, \zbt \mid \boldsymbol{\theta},  x_{1:T}\right)  d \ybm d \zbt \\
\intertext{adding condition \ref{cond:a}}
&\propto p(\tb) p(\pp) \int 
 p\left(\mbt^y \mid \ybo, \zbt, \boldsymbol{\phi}, x_{1:T} \right) d \zbt   \\
& \hspace{1cm} \times \int\int p\left(\ybo, \ybm, \zbt \mid \boldsymbol{\theta},  x_{1:T} \right) d \ybm d \zbt  \\
\intertext{and \ref{cond:b}}
&= p(\tb) p(\pp) p\left(\mbt^y \mid \ybo, \pp, x_{1:T} \right) \int \int  p\left(\ybo, \ybm, \zbt \mid \boldsymbol{\theta},  x \right) d \ybm d \zbt \\
&= p(\pp) p\left(\mbt^y \mid \ybo, \pp, x_{1:T} \right) \times p(\tb) L \left(\boldsymbol{\theta} \mid \ybo,  x_{1:T} \right). 
\end{align*}
so that the posterior factors over the missing data model and the state space model.
The crucial result here is that, in order to satisfy ignorability, it is not sufficient for the \textit{missing mechanism to be independent of the unobserved emissions}, but it also needs to be \textit{independent of the unobserved latent states}. 

Note that conditioning on the latent $\zbt$ in Condition \ref{cond:a} is reminiscent of  ``Auxiliary-MAR'' formulated in \cite{daniels2008missing}, the key difference here being that $\zbt$ is entirely latent, necessitating further assumptions about its relationship with $\mbt$ in Condition \ref{cond:b}. 
Finally, note that, owing to the completely unobserved nature of $\zbt$, an alternative formulation of ignorability could treat $\zbt$ as an additional parameter considered as part of $\tb$, $\tb(\zbt)$. 
Condition \ref{cond:b} in that case would have been implied by an extended version of condition \ref{cond1}, which would have required $p(\mbt^y \mid \pp, \tb(\zbt)) = p(\mbt^y \mid \pp)$.
We have intentionally formulated the ignorability as in Definition 1 in order to support more interpretable assessment of the required missing data assumptions in practice.  
Finally, note that the ignorability condition was formulated to support posterior inference for the parameters of the state space model, but an immediate consequence is the ability to impute missing values for $\ybm$ directly from their posterior predictive distribution. 
A more formal statement about this aspect of the model is given in Web Appendix A.

Figure \ref{fig:dag} outlines Condition \ref{cond:3} in terms of directed acyclic graphs (DAGs) depicting relationships among covariates, latent state indicators, missing data indicators, and outcomes that might be modeled with a state space model. 
Black arrows correspond to relationships that can be specified in the state space model and red arrows indicate relationships that are considered in Conditions \ref{cond:a} and \ref{cond:b}.
For example, in the motivating PA study, latent activity is informed by covariates ($x_{1:T} \rightarrow \zbt$) and dictates PA outcomes ($\zbt \rightarrow \ybm$), and whether a PA value is observed can generally depend on covariates ($x_{1:T} \rightarrow \mbt$), the latent activity state ($\zbt \rightarrow \mbt$), and the realized PA value ($\ybm \rightarrow \mbt)$. 
In some contexts, covariates may have additional relationships with emissions ($x_{1:T} \rightarrow \ybm$).  
For the purposes of illustration, we have omitted additional possible dependencies on $\ybo$.
The threat of nonignorable missing data derives from possible dependencies among latent state, the missing data indicator, and outcomes, which would manifest in the DAG as any $\mbt - \ybm$ association.

\begin{figure}[h!]
    \centering
    \begin{tikzpicture}
    
    \draw[thick] (0,0) rectangle (18,5.2); 

    \draw[thick] (0.5,0.5) rectangle (3.5,4.5); 
    \node at (2,4.8) {\textbf{I}};
    \node at (2,4) {$x_{1:T}$};
    \node at (1,2.3) {$\zbt$};
    \node at (3,2.3) {$\mbt^y$};
    \node at (2,0.8) {$\ybm$};
    \draw[thin] (1.7,3.7) rectangle (2.3,4.3); 
    \draw[->] (2,3.8) -- (1.1,2.5); 
    \draw[->] (2,3.8) -- (2,1.1); 
    \draw[->] (2,3.8) -- (2.8,2.5); 
    \draw[->] (1.1,2.1) -- (1.85,1.1); 
    \draw[->, red, thick] (1.5,2.25) -- (2.5,2.25); 
    \draw[->, red, thick] (2.1,1.1) -- (2.8,2.1); 
    
    \draw[thick] (4,0.5) rectangle (7,4.5); 
    \node at (5.5,4.8) {\textbf{II}};
    \node at (5.5,4) {$x_{1:T}$};
    \node at (4.5,2.3) {$\zbt$};
    \node at (6.5,2.3) {$\mbt^y$};
    \node at (5.5,0.8) {$\ybm$};
    \draw[thin] (5.2,3.7) rectangle (5.8,4.3); 
    \draw[->] (5.5,3.8) -- (4.6,2.5); 
    \draw[->] (5.5,3.8) -- (6.3,2.5); 
    \draw[->] (4.6,2.1) -- (5.35,1.1); 
    \draw[->] (5.5,3.8) -- (5.5,1.1); 
    \draw[->, red, thick] (5,2.25) -- (6,2.25); 

    \draw[thick] (7.5,0.5) rectangle (10.5,4.5); 
    \node at (9,4.8) {\textbf{III}};
    \node at (9,4) {$x_{1:T}$};
    \node at (8,2.3) {$\zbt$};
    \node at (10,2.3) {$\mbt^y$};
    \node at (9,0.8) {$\ybm$};
    \draw[thin] (8.7,3.7) rectangle (9.3,4.3); 
    \draw[->] (9,3.8) -- (8.1,2.5); 
    \draw[->] (9,3.8) -- (9.8,2.5); 
    \draw[->] (8.1,2.1) -- (8.85,1.1); 
    \draw[->] (9,3.8) -- (9,1.1); 

    \draw[thick] (11,0.5) rectangle (14,4.5); 
    \node at (12.5,4.8) {\textbf{IV}};
    \node at (12.5,4) {$x_{1:T}$};
    \node at (11.5,2.3) {$\zbt$};
    \node at (13.5,2.3) {$\mbt^y$};
    \node at (12.5,0.8) {$\ybm$};
    \draw[thin] (12.2,3.7) rectangle (12.8,4.3); 
    \draw[->] (12.5,3.8) -- (11.6,2.5); 
    \draw[->] (12.5,3.8) -- (13.3,2.5); 
    \draw[->] (11.6,2.1) -- (12.35,1.1); 

    \draw[thick] (14.5,0.5) rectangle (17.5,4.5); 
    \node at (16,4.8) {\textbf{V}};
    \node at (15,2.3) {$\zbt$};
    \node at (17,2.3) {$\mbt^y$};
    \node at (16,0.8) {$\ybm$};
    \draw[->] (15.1,2.1) -- (15.85,1.1); 
    \draw[->, red, dashed, thick] (15.5,2.25) -- (16.5,2.25); 

    \end{tikzpicture}
    \caption{Different scenarios regarding the relationship between $\mbt$, $\ybm$, $x_{1:T}$, $\zbt$.}
    \label{fig:dag}
\end{figure}
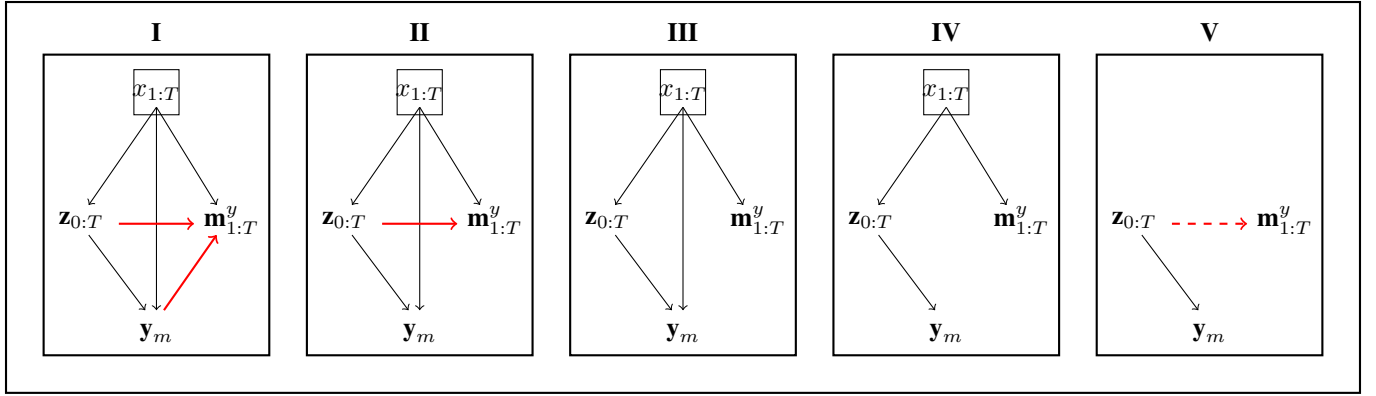

In the fully connected DAG in I, ignorability is clearly violated due to the direct dependence between $\mbt \rightarrow \ybm$. 
The DAG in II adopts Condition \ref{cond:a}, removing the arrow from $\ybm$ to $\mbt$.  
In a PA study, the presence of $\ybm \rightarrow \mbt$ would indicate that whether a value of $\ybt$ is measured is related to its value, independently of the latent activity state. 
An example might be device malfunction where, for example, a heart rate measure is never stored whenever the BPM has zero as second digit. 
We regard such instances as purely technical and do not consider violations of Condition \ref{cond:a} a major threat in studies of PA. 
The key point is that, even when Condition \ref{cond:a} holds as in II,  the threat to ignorability remains through the $\ybm \leftarrow \zbt \rightarrow \mbt$ path, since $\zbt$ is entirely latent.  
Thus, a latent dependence between $\ybm$ and $\mbt$ persists.   
The DAG in III retains Condition \ref{cond:a} and adds condition \ref{cond:b} through deletion of the arrow from $\zbt \rightarrow \mbt^y$. 
In a PA study, the $\zbt \rightarrow \mbt$ would indicate that whether a value of $\ybt$ is measured is related to the latent activity state.  An example might be that a person is more likely to wear their device when exercising.  We regard this type of behavior-related missingness as the salient threat in wearable device PA studies.  
In this case, the conditioning on the observed covariates $x_{1:T}$ blocks the only remaining backdoor path between $\ybm$ and $\mbt^y$, precluding any remaining association between $\ybm$ and $\mbt$ and satisfying the ignorability conditions.  This underscores the potential importance of being able to express a dependence between $x_{1:T}$ and $\zbt$ when specifying state space models.



\subsection{Implications for specifying state space models and limits of the HMM}
The need to explicitly condition on $x_{1:T}$ to satisfy Condition \ref{cond:3} can manifest differently in different state space model specifications.  
For example, in a Gaussian dynamic model such as the one in \cite{cai2022state} deployed to model human behavior, the latent states are not intended to hold any interpretation and covariates relevant to missing data can be included in the emission distribution.  Satisfaction of Definition \ref{def} follows because conditioning on $x_{1:T}$ in the emission distribution can satisfy Condition \ref{cond:a}, with Condition \ref{cond:b} holding without the need to condition on $x_{1:T}$. 
In contrast, settings where the underlying process of interest is conceived as transitions among discrete and interpretable latent states, an HMM may be appropriate. For an HMM, additional dependence between covariates and $y_t|z_t$ in the emission distribution would imply a relationship that was unrelated to the latent state.  In a study of PA where the latent states are activity classifications, this may relate to device malfunction, but could not relate to PA without distorting the meaning of $\zbt$. As these are not the most salient threat in the study of PA, we omit this type of dependence in DAG (IV) of Figure \ref{fig:dag}, which removes the direct arrow from $x_{1:T} \rightarrow \ybm$.  Thus, the settings when an HMM is appropriate for modeling PA are among the settings where, if covariates are thought to relate to PA and missingness, Condition \ref{cond:b} requires conditioning on $x_{1:T}$ when modeling the latent states.  DAG V in Figure \ref{fig:dag} depicts the implications of a failure to encode a $x_{1:T} \rightarrow \zbt$ relationship, yielding a residual association between $\zbt$ and $\mbt$, thus violating the ignorability conditions.
 


\subsection{Isolating deficiencies of HMMs with missing data}
For an HMM as specified in Section \ref{s:HMM} with emission distribution $p_{\psi}(\yb_t \mid z_t)$ and a transition distribution as in Equation (\ref{eq:HMM2}), Condition \ref{cond:a} implies that $p(\mbt \mid \ybm , \zbt) = p(\mbt \mid \zbt)$, and
\begin{equation}
\label{eq:LfullHMM}
L ( \boldsymbol{\theta}, \boldsymbol{\phi} \mid \ybo, \mbt^y,  x_{1:T})
=  \int \prod_{t=0}^{T} p_{\pp}(m_t \mid z_t) \int p_{\pi}(z_0) \prod_{t=1}^{T} q_{z_{t-1},z_{t}} p_{\psi}(y_t \mid \theta_{z_t})d \ybm   d \zbt 
\end{equation}
with the implied relationship between $\zbt$ and $\mbt^y$
not allowing a partition between the PA model and the missing data model in the posterior.
To satisfy Condition \ref{cond:b}, $x_t$ could be included in the model for the transition distribution by augmenting Equation (\ref{eq:LfullHMM}) to become : 
\begin{align*}
&L ( \boldsymbol{\theta}, \boldsymbol{\phi} \mid \ybo, \mbt^y,  x_{1:T}) \\
& = \prod_{t=0}^{T} p_{\pp}(m_t \mid x_t) \int  \int p_{\pi}(z_0) \prod_{t=1}^{T} q_{z_{t-1},z_{t}}(x_t) p_{\psi}(y_t \mid z_t) d \zbt d \ybm \\
& =  p(\mbt \mid x_{1:T}, \pp) \times  L \left(\boldsymbol{\theta} \mid \ybo ,  x_{1:T} \right).
\end{align*}
which is an important deviation from the HMM expressed in Section \ref{s:HMM}. Thus, classical HMMs
lack the capacity to properly incorporate covariates $x_t$, into the transition distribution, presenting an inherent limitation for including covaraites covariates required to satisfy ignorability Condition \ref{cond:b}.  In the context of a PA study, this translates to an inability to account for features such as the hour of the day to account for daily routine, as has been considered in previous PA studies \citep{langrock2013combining, st2021use, witowski2014using}).


\section{Nonhomogeneous HMMs for modeling PA}
\label{s:NHMM}

\label{HMM:num2}
Extending HMMs to include covariates in the transition distribution corresponds to a Non Homogeneous NHMM, where transition probabilities between states depend on the level(s) of covariate(s).  Extensions to NHMMs mainly use multinomial regression to model transition probabilities, but the computational burden of such models has historically limited their application to small numbers of observations \citep{altman2007mixed}, settings where the underlying states are not of interest \citep{lu2023bayesian}, or otherwise avoided Bayesian inference \citep{maruotti2012mixed, huang2018hidden}.
Bayesian implementations in particular have always been considered exceptionally challenging due to the lack of closed form for posterior distributions, but recent solutions have recently emerged based on the Pólya-Gamma data augmentation strategy \citep{polson2013bayesian} to sample from the conditional posterior of parameters in the transition distribution \citep{holsclaw2017bayesian, wang2023bayesian}.
We propose  a model that makes use of such sampling strategies and accommodates the aforementioned ignorability implications of the HMM in the context of PA tracking data with missing observations.

\subsection{NHMM extension: inference via Pólya-Gamma data augmentation and model specification}
\label{NHMM}
We extend the previous expression of the HMM to the following NHMM likelihood:
\begin{equation}
\label{eq:likNHMM}
    p(\ybt, \zbt) = p_{\pi}(z_0) \prod_{t=1}^{T} q_{z_{t-1}, z_{t}}(x_t) N(y_t \mid z_t, \psi),
\end{equation}
with the crucial extension being the expression for the transition probabilities, $p_q(z_t \mid z_{t-1}, x_t) = q_{z_{t-1}, z_t}(x_t)$.
Denote by $X$ the $T \times p$ design matrix containing covariates, and by $x_t$ its $t^{th}$ row.
We define\begin{equation}\label{eq:entries}
  q_{ij}(x_t) = P\left(z_t=j \mid z_{t-1}=i,  x_t \right)=  \frac{\exp \left(\xi_{i j}+x_t^{\prime} \boldsymbol{\beta}_j\right)}{\sum_{m=1}^K \exp \left( \xi_{im}+  x_t^{\prime} \boldsymbol{\beta}_m\right)}.
 \end{equation}
Note that such specification for $q_{ij}$ will not only allow us to flexibly model $Q$, but it will allow for a model that permits inference on the parameters in (\ref{eq:entries}), $\boldsymbol{\xi}_{ij}$ and $\bb_j$, which might be of interest  as descriptors of how people transition between latent PA categories according to the relevant covariates. The parameter $\bb$ in (\ref{eq:entries}) relates the covariate $x'_t$ to the probability of transitioning {\it to} state $j$.
While the model could in principle accommodate any covariate of interest and patients characteristics, we specify $x'_t$ to encode information about clock time, captured by 23 dummy variables corresponding to the 24 hours of the day, to capture patients' routine activity patterns. 
Note the lack of dependence of this relationship on the current state, $i$, implying that the hour of the day impacts the probability of transitioning into an activity state in a manner that does not depend on the originating state. The role of the incoming state is captured by $\xi_{ij}$, which serves a role akin to an $i$-indexed intercept. 
This is a distinction from other similar models and Bayesian posterior sampling strategies deployed in other settings \citep{wang2023bayesian} and reflects our interest in how daily routine impacts activity patterns.

A posterior sampling scheme for $\zz_j = (\xx_j, \bb_j)$ using the Pólya-Gamma augmentation strategy \citep{polson2013bayesian, holsclaw2017bayesian}, is described in Web Appendix B.
Inference is performed by sampling from the full-data response model, achieved through a data augmentation step in which the missing data are also sampled \citep{daniels2008missing}. 
It is worth noting that, with a model satisfying ignorability conditions, the factorial nature and dependence structure of the likelihood would have enabled the observed-data likelihood to be obtained without the need to approximation through missing data imputation \citep{zucchini2009hidden, daniels2008missing}. 
However, the data augmentation step proves particularly useful given our interest in performing data imputation.
Details of the Gibbs sampling steps are provided in Web Appendix C.

We set the number of states $K$ equal to 3, corresponding to a scenario in which individuals could be in a sedentary, intermediate, or high activity state.
Denoting by $p$ the number of covariates included in the design matrix $X$, all the parameters that enter in the logistic link in (\ref{eq:entries}) are collected in a $K \times (K+p)$ matrix, so that each $k$ row corresponds to the parameters that are directly responsible of the probability of going in the $k$-th state, i.e. $\xi_{ik}$ and $\beta_k$ as in (\ref{eq:entries}). 
The $K$-th row of this matrix is set to zero to preserve identifiability (see \cite{agresti2012categorical}), so that the remaining parameters should be interpreted in terms of the probability of transitioning into state $j$ relative to transitioning into the  \textit{basilene state}. 

For implementation, an NHMM is specified and learned separately for each individual, which can accommodate individuals with very different PA habits.  Prior distributions for the emission distribution parameters are specified to be non-informative.
Informative priors are specified for $\zz$ in order to facilitate MCMC convergence, as outlined in Section \ref{app:NHMM3}.

\subsection{Practical Considerations}
\label{app:NHMM3}
As with many state-space models, posterior inference for NHMMs presents nontrivial considerations for model identifiability and MCMC convergence. We identify the notion of {\it separation} as a key determinant of complications, where separation in our PA case refers to certain hours of the day where a person never exhibits transitions to or from a certain latent state. This represents a specific instance of the more general documented problem of poor MCMC mixing when Pólya-Gamma data augmentation schemes are applied to categorical data with imbalanced categories
\citep{johndrow2019mcmc}.
Such situation can easily arise in PA studies, particularly those that include overnight hours where certain activity levels are rare.
Intuitively, this will create a lack of information for learning the parameters that dictate the transition probabilities, i.e. $\zz$.
Problems with MCMC convergence and sample autocorrelation are exacerbated in the presence of missing data. 
Identifiability and MCMC convergence issues for parameters analogous to $\zz$ receive limited attention in previous work in Bayesian NHMMs such as \cite{holsclaw2017bayesian} and \cite{wang2023bayesian}.

We identify several practical considerations for effective MCMC to approximate the posterior distribution of model parameters in (\ref{eq:likNHMM}) - (\ref{eq:entries}) in the PA context.
The first is strategic choice of the hour of the day that will serve as the baseline category, which has most direct bearing on mixing for the intercept parameters of the transition probabilities, $\xi_{ij}$.
Choosing an hour unlikely to contain transitions between states (e.g., 12:00am) as the baseline hour can severely challenge MCMC mixing of $\boldsymbol{\xi}$ and impact the estimation of $\boldsymbol{\beta}$ at all hours.
Choosing a baseline hour more likely to show transitions between categories (i.e., with lower threat of separation) will limit such convergence issues.
The second strategy we employ to improve MCMC performance is regularization through informative priors on $\zz$. 
In particular, we set the prior distribution to $\boldsymbol{\zeta}_j \sim N_{K+p}(\boldsymbol{\zeta}_0, I \cdot (1/10))$. 
Finally, we consider an additional regularization strategy based on a data augmentation prior \citep{greenland2001data} for the marginal probability of membership in each state, $\sum_{i}q_{ij}(t)$ for $j=1,2,3$.  
This prior is specified by augmenting the observed data with synthetic pseudo-data representing $m \times K$ days of specified values of $(y_t, z_t)$, with $z$ chosen to entail equal representation for states $1,..., K$ at every $t$ and $y$ values set according to values of $\psi$ learned after a number of exploratory (and subsequently discarded) MCMC iterations. 
$m$ determines the strength of the prior, and we found in our analysis AYA cancer patients that  
specifying $m = 1$ was typically effective at stabilizing estimates without strongly influencing the posterior for most patients for whom MCMC convergence was a pressing consideration. In general, the data augmentation prior was most useful for patients with high missing percentages of missing data. Web Appendix D provides additional details on the implementation of this strategy.

\section{Simulation Study}
\label{s:app} 
To compare the HMM and NHMMs for accommodating missing data with wearable device PA data, we conducted a realistic simulation study based on observed data from selected patients from the AYA cancer cohort, with artificial missingness simulated according to different scenarios. 
We identify five patients from the AYA cohort showing a negligible percentage of missing PA measures. We simulate $\mbt$ according to a $Ber(\boldsymbol{p}_\ell)$ distribution for $\ell = 1, ..., 24$ hours of the day, where each $t=1, ...,T$ is associated to a corresponding $\ell$ and the probability of missingness depends on the hour of the day via: $
\boldsymbol{p}_{\ell} = (1-\gamma) \cdot \boldsymbol{p}_{0,\ell} + \gamma \cdot \boldsymbol{p}_{h,\ell}
$. The vector $\boldsymbol{p}_0 = \boldsymbol{p}_{0,1}, ..., \boldsymbol{p}_{0,24}$ is constant across the hours of the day, while the vector $\boldsymbol{p}_h = \boldsymbol{p}_{h,24}, ..., \boldsymbol{p}_{h,24}$ entails values that are different across different hours of the day.  Hence, specifying $\gamma=0$ corresponds to the case of MCAR, whereas values of $0<\gamma\le1$ simulate varying degrees of $x \to \mbt$ dependence to simulate that missigness depends on the hour of the day.  

Figure \ref{fig:combined} compares results from the HMM and the proposed NHMM across different missing data scenarios for two individuals, here identified as ID 8 and ID 16.  The Figure depicts the median measure of outperformance - or underperformance - of the NHMM relative to the HMM.
The value of $\gamma$ used for each simulation setting is reported on the $x$-axis; higher values of $\gamma$ correspond to larger deviations from data being MCAR. the higher the role of the hour of the day in the imposed missing pattern.
Two missing percentages were used, reflecting two different degrees of missingness that were observed in the AYA cohort.

The results reported in the figures show that the two models perform similarly when there is little or no dependence on $x$, which corresponds to the case where a MCAR assumption (approximately) holds. The NHMM clearly performs better as the dependence on $x$ increases, with more marked improvements when the overall percentage of missing observations is higher.
Note that the simulation controls the relationship between $x \to \mbt$, but does not control the relationship between $x \to \zbt$, which would correspond to the AYA individual's actual daily PA routine. As shown in Web Appendix E, individuals ID 8 and ID 16 show evidence of different routine patterns,indicating that the result apparent from Figure \ref{fig:combined} is not due to an exceptionally high $x \to \zbt$ relationship specific to the selected individuals.
The remaining individuals used in this hybrid simulated dataset showed very comparable results, with a visualization of their combined output also provided in the Web Appendix E. 
\begin{figure}[htbp]
\centering
\begin{tabular}{c|c|c}
   & \textbf{ID8} & \textbf{ID16} \\ \hline
   \rotatebox{90}{\textbf{Medium missing \%}} &
   \includegraphics[width=0.38\textwidth]{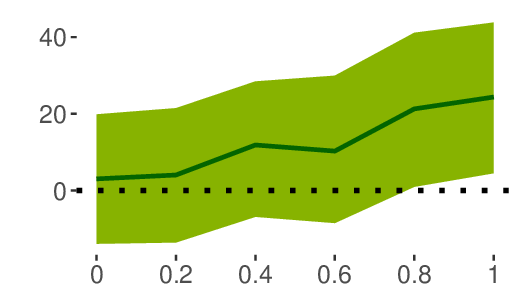} &
   \includegraphics[width=0.38\textwidth]{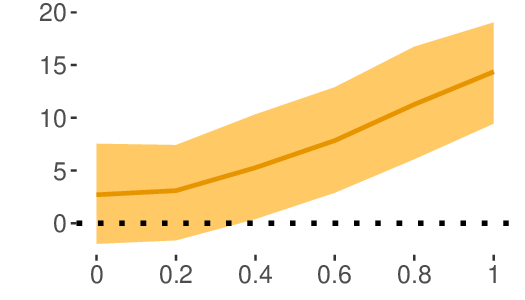} \\[0.5cm]
   \rotatebox{90}{\textbf{High missing \%}} &
   \includegraphics[width=0.38\textwidth]{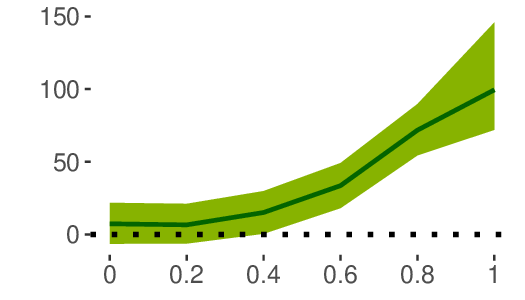} &
   \includegraphics[width=0.38\textwidth]{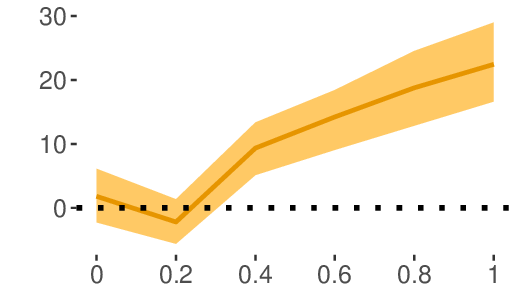} \\
\end{tabular}

\caption{Comparing HMM and NHMM under different simulated scenarios captured by $\gamma$ and different individuals, using RMSE for predictions of artificially-induced missing data. The y-axis represents the difference in performance between the two models; the higher the value, the better the performance of the NHMM relative to the HMM. The x-axis values are the values of $\gamma$ used to induce the missingness pattern. Lines are posterior medians and bands are 90\% credible intervals.}
\label{fig:combined}
\end{figure}

Web Appendix F contains a more controlled simulation study that dictates both the $x_{1:T} \to \mbt$ and $x_{1:T} \to \zbt$ relationships with entirely generated data. 
In those simulations the NHMM outperforms the HMM in terms of missing data imputation and parameter estimation across stronger simultaneous relationships $x_{1:T} \to \mbt$ and $x_{1:T} \to \zbt$.

\section{Analysis of AYA Cancer Patients}
\subsection{AYA Cancer Patients Dataset}
\label{app:num0}
PA measures were collected by the MD Anderson Cancer Center as part of a broader study that provided Fitbit Charge 3 devices to a cohort of AYAs who were between 15 and 19 years of age and undergoing cancer treatment. Enrolled patients agreed to participate in a series of studies whose primary objective was to identify biomarkers of cardiac impairment associated with receipt of cardiotoxic chemotherapy. 
Additional objectives included understanding the evolution of PA in individuals undergoing such treatments and to record their activity for a span of time up to three years. Fitbit collects minute level data for heart rate and steps.
We will use those categories only for comparison, since our methods are designed for a more general class of devices, which may not provide such classification by default.  
What's more, Fitbit still delivers a value of 1 (sedentary) for the latent activity category whenever there is an \texttt{NA} in steps or heart rate, hence in observations that should instead be categorised as missing.

We use heart rate and steps as $\ybt$ and apply the following pre-processing to 58 patients available for analysis. First, we apply an established algorithm \citep{choi2011validation} designed to ensure that data are classified as missing when the device is not worn and to avoid anomalous measures in case in which the device is not worn but still recording data.
Then, we aggregate data to 15 minute intervals for dimension reduction purposes, so that each observation $\textbf{y}_t$ corresponds to a 15 minute period.
Each 15 minute window with at least 10 minutes of missing observations is labeled as missing, i.e. $m_{t} =1$. 
For the remaining 15 minute windows, $m_{t} = 0$ and $\textbf{y}_t$ is set as the sum of the observed PA outcomes in the interval, possibly rescaled to account for any missing minutes (in this case less than 10) in such window.

We then discard all $t$ falling in days that contain less than 5 hours of observations between 8AM and 8PM, which we regard as days when the patient is not monitoring PA and for which there is not enough information in a 24 hour period to reliably recover activity patterns. Overall, the goal is to discard data for days during which individuals are not actively participating in the PA study while retaining as many participating days as possible to learn physical activity patterns through imputation of missing observations. We retain only the patients that have at least 30 days of monitoring. Note indeed that we still consider nighttime hours as part of our study, accommodating for example the need to estimate quality of sleep or individual nighttime routine. 

The final analysis dataset consists of 36 individuals with an average length $T$ of 121578 15 minute intervals, corresponding to a mean of roughly 127 days of PA monitoring. The percentage of missingness spans a minimum of 0.01 to a maximum of 0.42.  Each patient is analyzed separately with the HMM and proposed NHMM as outlined in Sections \ref{s:HMM} and \ref{s:NHMM} with $K-3$ latent activity states corresponding to sedentary, moderate, or high physical activity.  
 
\subsection{Results}
\label{app:num4}
The models are used to estimate parameters of the NHMM, impute missing PA observations, and calculate interpretable PA summaries. We show in particular how such summaries, which are commonly used to summarize patient PA in contexts such as the AYA study, can differ depending on whether the underlying modl is the HMM, the NHMM, or even just the raw Fitbit output is used. To illustrate, we use the estimated marginal probability of being in an activity state depending on the hour of the day, the average daily total step count, the average HR per wear minute, the probability of being in the sedentary state and the average length of sedentary bouts.
Results were obtained by running the MCMC chains for 20000 iterations with 15000 sweeps as burn-in. 
Convergence of the MCMC chains was assessed by visual inspection of traces for all the parameters in the model, considered separately for each patient. With only the carefully chosen baseline hour and regularization prior on $\boldsymbol{\zeta}_j$ specified in Section \ref{app:NHMM3}, only 9 patients exhibited adequate MCMC convergence. 
For the remaining patients, we specified the data augmentation prior on the marginal probability of state membership as also explained in \ref{app:NHMM3}. Values of $m=1$ or $m=2$ in the data augmentation prior were sufficient to achieve adequate MCMC performance for 20 patients. 
7 patients required larger values of $m$ to stabilize MCMC performance, with $m=7$ the largest value used.  
A sample of MCMC trace plots is shown in Web Appendix G, together with the corresponding traces obtained without the data augmentation prior for comparison. 
 
Figure \ref{fig:marginals} summarizes the comparison of NHMM and HMM estimated marginal probabilities of being in each of the states for each of the 24 hours of a day for a single patient. As one can easily see, the HMM indicates similar probabilities of activity state membership during the daytime and nighttime hours. 
In contrast, the NHMM, which models transitions among PA states according to the hour of the day, more realistically estimates that the probabilities of membership in the activity states evolves across the morning, afternoon, and night. 
Estimating different probabilities of latent PA state during different times of day suggests benefits for missing data imputation. 
\begin{figure}[h]
\centerline{%
\includegraphics[width=0.5\linewidth]{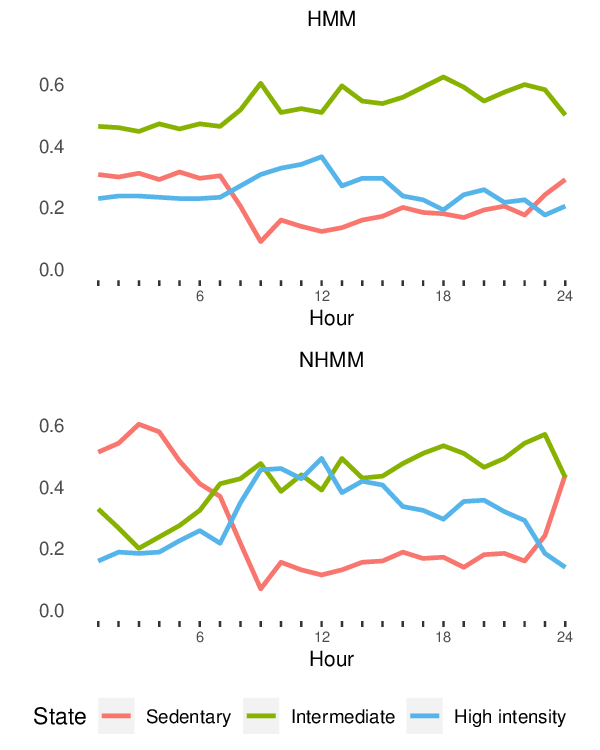}}
\caption{Posterior marginal probability of being in a  state.}
\label{fig:marginals}
\end{figure}

Figure \ref{fig:results} shows the remaining illustrative metrics for a representative set of individuals.
Posterior samples for the summaries are shown for the NHMM and HMM models, respectively in green and orange.
For the first two metrics - average total daily steps and average HR per minute -  estimates from raw Fitbit data with no missing data imputation are shown in blue for comparison.
The average step count was obtained by summing all the steps recorded in a day and then averaging the results across days. 
This implies that reliance on raw Fitbit with no imputation reports a summary that is biased towards lower totals due to its use of fewer observation times.
The HR per minute average was instead obtained averaging over the data within each day. 
By averaging over the observed data only, Fitbit generally reports higher estimates, this being in line with the patients being inclined to wear their device in periods of higher activity.
The comparison among NHMM and HMM shows higher estimates for steps and heart rate for the HMM in both cases.
This is also consistent with people being more likely to actually wear their device during periods of activity, with the HMM imputing values of steps and HR that are dominated by the bulk of the day when patients are not active and the NHMM inferring that periods of higher activity are more likely during certain hours and imputing values accordingly.
The posterior probability of being in the sedentary state and the average sedentary bout length confirm this difference between the results from the two models, coherently again with what shown in Figure \ref{fig:marginals}.
The HMM provides an estimate that flattens the states and values throughout the day, thus limiting its capability to capture the relatively less represented states in observed the data.
\begin{figure}[htbp]
  \centering
  \begin{minipage}{0.24\textwidth} 
    \caption*{Avg total steps per day
    } 
    \includegraphics[width=\linewidth]{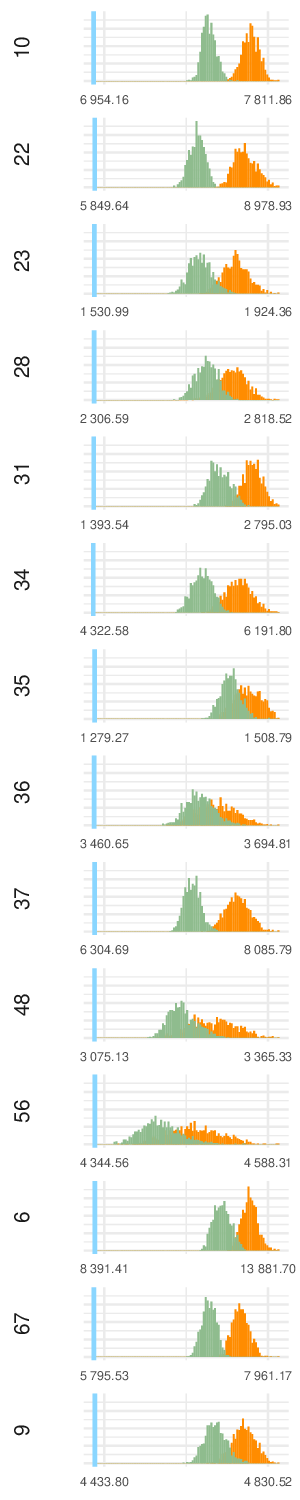}
  \end{minipage} 
  \begin{minipage}{0.24\textwidth}
    \caption*{Avg HR per minute} 
    \includegraphics[width=\linewidth]{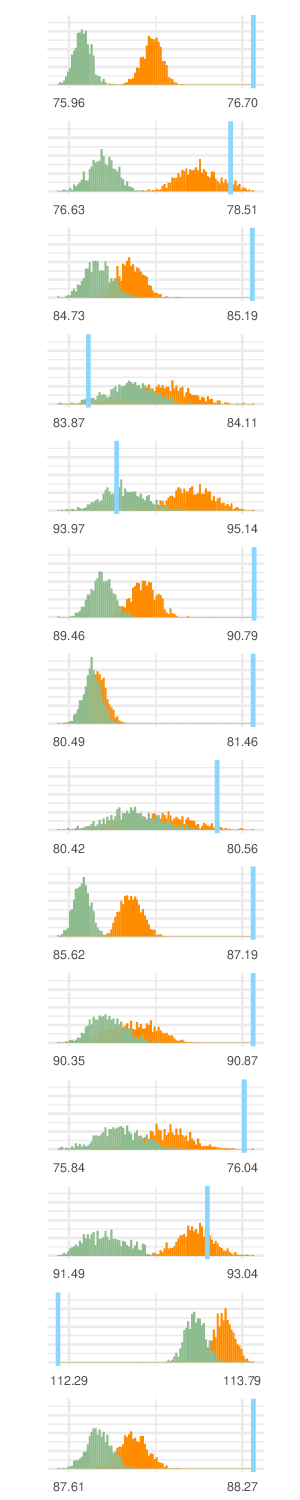}
  \end{minipage}
  \begin{minipage}{0.24\textwidth}
    \caption*{Prob. of sedentary state}     \includegraphics[width=\linewidth]{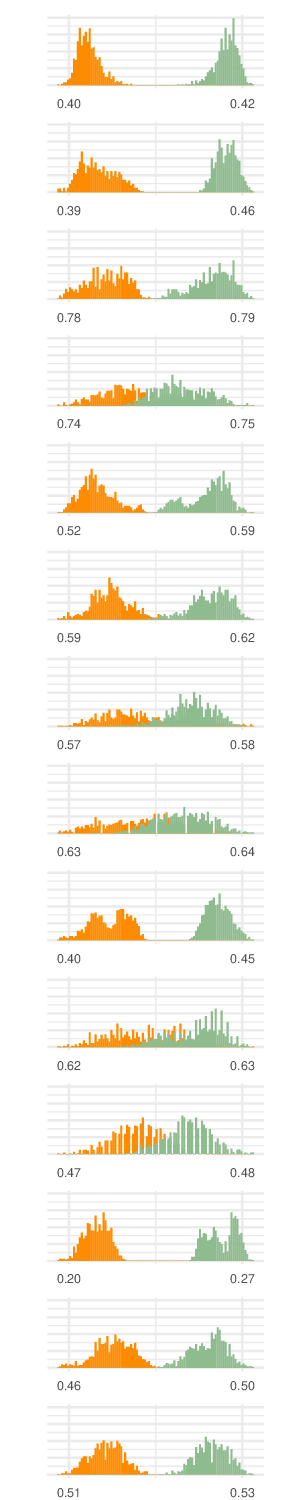}
  \end{minipage}
  \begin{minipage}{0.24\textwidth}
    \caption*{Avg night sedentary bout} 
    \includegraphics[width=\linewidth]{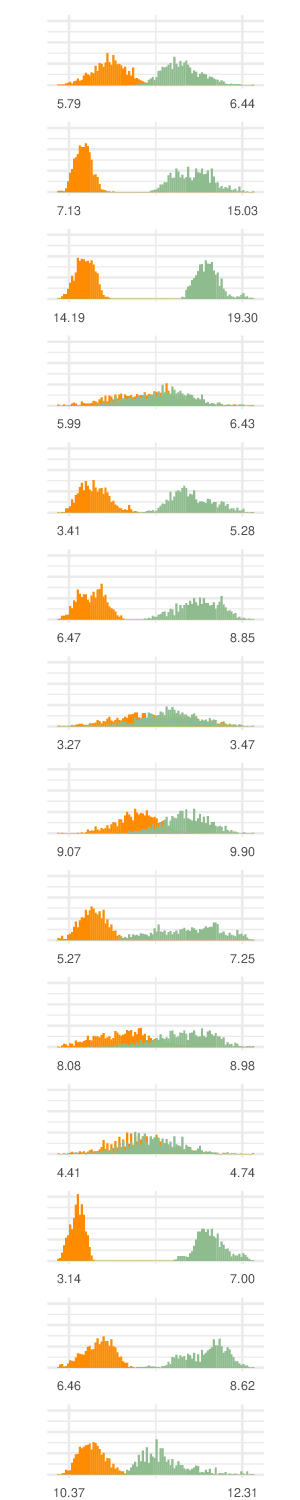}
  \end{minipage}
  \caption{Posterior samples of quantities of interest. Green refers to posterior estimates from the NHMM model, orange refers to posterior estimates from the HMM model, and blue to estimates obtained from raw Fitbit data. The bout length is expressed in hour quarters and corresponds to night hours.}
  \label{fig:results}
\end{figure}

\section{Conclusion}
\label{s:discuss}
We provide a formal consideration of missing data and ignorability in state space models, focusing on models for PA measured from wearable devices.
In particular, we focus on Bayesian HMMs and highlight the inability of common HMMs to make reliable inferences for and imputation of PA data in the presence of non-ignorable missingness.  In response, we offer a Bayesian NHMM to accommodate covariates in the model for transitions between latent states, and show how this permits satisfaction of ignorability conidtions when covariates relating to the missigness mechanism and underlying PA are available.
This expansion of the traditional HMM and attendant Bayesian implementation leads directly to improved imputation of missing data and quantification of uncertainty, which are two improvements over the available activity categorization and summaries provided by Fitbit.  We also highlight several practical implementation challenges relating to MCMC performance that have not received detailed attention in works employing similar NHMM implementations.  We provide several strategies, including a data augmentation prior on the marginal distribution of latent states, that showed potential to alleviate MCMC convergence issues in this type of model. In the context of a PA study of AYA cancer patients, we showed the potential for more reliable inferences and missing data imputation with our Bayesian NHMM compared to a standard HMM or a complete-case analysis.  The analysis of AYA cancer patients highlights the need to consider missing data mechanisms more complex than the commonly (and often implicitly) invoked MCAR and that modeling advances that can support the weaker MAR assumption by appropriately including covariates can produce practically meaningful differences in PA summaries.  

Future research could continue and improve this line of work.  For example, more general considerations for imbalanced categorical data such as those in \cite{johndrow2019mcmc} could possibly indicate MCMC strategies for NHMMs that substitute the Pólya-Gamma data augmentation with simpler Gibbs steps in presence of imbalanced latent activity categories.  Furthermore, improvements in computation might generate new directions in hierarchical modeling of individual PA data that borrow information across patients to generate population inferences and further improve missing data imputation. Information pooled among individuals with similar PA patterns might also serve to compensate for lack of information or complete separation in some patients and improve MCMC performance.  Finally, our analysis fixed the number of latent states to $K=3$ to correspond to common categorizations of physical activity.  Further exploration could estimate a different number of states for different individuals, possibly through reversible jump strategies.

\section*{Acknowledgments and Supplementary Material}
This work is partially supported by a grant from the Cancer Prevention and Research Institute of Texas (RP200381).
We thank the MD Anderson Cancer Center for providing us with the data, and Michelle Audirac for the preliminary work conducted on Fitbit data preprocessing.
Web Appendix A, B, C, D, E, F, G are available online.
The code for implementing our model is located on the author's Github page  Bayesian-NH-HMM-for-PA-Data.
\label{lastpage}

%% file: supplementary_material.tex
\newpage

\section*{Web Appendix A: Ignorability and missing data imputation}
One notable advantage of the Bayesian approach in the context of missing data is that, under the ignorability assumption, posterior estimation based on observed data directly yields the posterior predictive distribution of 
$\ybm$.
This facilitates the imputation of missing values while naturally incorporating uncertainty, thereby providing a robust framework for handling incomplete data.
\begin{proposition}
\label{pro2}
For a state space model with outcome $\ybt$, available observations $\ybo$, missing observations $\ybm$, states $\zbt$, parameters of interest $\boldsymbol{\theta}$, parameters governing the missing mechanism $\pp$, and  other observables $x_{1:T}$, posterior sampling of $\zbt$ and $\ybm$ can be done ignoring the missing mechanism when conditions given in Definition 1 in the main manuscript are satisfied.
\end{proposition}

\begin{proof}
    \begin{align*}
    p(\ybm, \zbt & \mid \ybo, \mbt^y, x, \tb, \pp), \\
    &= \dfrac{p(\ybm, \zbt, \ybo,  \mbt^y, x_{1:T}, \tb, \pp)}{p(\ybo, \mbt^y, x_{1:T}, \tb, \pp)}, \\
    &= \dfrac{p(\mbt^y\mid \ybm, \zbt, \ybo, x_{1:T}, \tb, \pp)   p(\ybm, \zbt,\ybo,  x_{1:T}, \tb, \pp)}
    {p(\mbt^y \mid \ybo, x_{1:T}, \tb, \pp)   p(\ybo, x_{1:T}, \tb, \pp)}, \\
\intertext{under condition (1) of Definition 1 from the main manuscript:}
    &= \dfrac{p(\mbt^y\mid \ybm, \zbt, \ybo, x_{1:T}, \pp)   p(\ybm, \zbt,\ybo,  x_{1:T}, \tb, \pp)}
    {p(\mbt^y \mid \ybo, x_{1:T},  \pp)   p(\ybo, x_{1:T}, \tb, \pp)}, \\
\intertext{under condition 3(a):}
    &= \dfrac{p(\mbt^y\mid \zbt, \ybo, x_{1:T}, \pp)   p(\ybm, \zbt,\ybo,  x_{1:T}, \tb, \pp)}
    {p(\mbt^y \mid \ybo, x_{1:T},  \pp)   p(\ybo, x_{1:T}, \tb, \pp)}, \\
\intertext{under condition 3(b):}
    &= \dfrac{p(\mbt^y\mid \ybo, x_{1:T}, \pp)   p(\ybm, \zbt,\ybo,  x_{1:T}, \tb, \pp)}
    {p(\mbt^y \mid \ybo, x_{1:T},  \pp)   p(\ybo, x_{1:T}, \tb, \pp)}, \\
    &\propto \dfrac{p(\ybm, \zbt,\ybo,  x_{1:T}, \tb, \pp)}
    {p(\ybo, x_{1:T}, \tb, \pp)}, \\
    \intertext{under condition (1)}
    &= p(\ybm, \zbt \mid \ybo, \tb).
\end{align*}
\end{proof}

\newpage
\section*{Web Appendix B: Derivation of posterior full conditional for $\boldsymbol{\zeta}_j$}
Denote by $X$ the $T \times p$ design matrix containing dummy variables indicating the hour of the day membership, with one hour-which we refer to as \textit{baseline hour}-taken out for identifiably purposes and by $X_t$ its $t^{th}$ row.
Denote by $Z$ the $T \times K$ binary matrix containing information about the \textit{incoming states}, meaning that the $t^{th}$ row of the matrix, $Z_t$,  takes value 1 in the column corresponding to $k$, where $z_{t-1} = k$, and zero elsewhere, so the transition distribution expressed in Equation (7) of the main manuscript 
can be written as
\begin{equation}\label{eq:entries2}
  q_{ij}(x) =   \frac{\exp \left(Z_t^{\prime}\xi_{ j}+X_t^{\prime} \boldsymbol{\beta}_j\right)}{\sum_{m=1}^K \exp \left( Z_t^{\prime}\xi_{m}+  X_t^{\prime} \boldsymbol{\beta}_m\right)}
 \end{equation}
with the parameters for state $K$ set to 0 for identifiability purposes.
The full conditional of interest is
\begin{align*}
    p&(\boldsymbol{\beta}, \boldsymbol{\xi} \mid \ybt, \zbt, \psi)  \propto p( \zbt \mid \boldsymbol{\beta}, \boldsymbol{\xi}) \cdot p(\boldsymbol{\beta}, \boldsymbol{\xi}) \\
    &= \prod_{k=1}^{K} p(\bb_k, \xx_k) \times \prod_{t=1}^{T} \prod_{k=1}^K \left[ \dfrac{exp \{  \textbf{Z}_t \xx_k + \textbf{X}_t \bb_k \}}{1+ \sum_{j=1}^{K-1}exp \{ \textbf{Z}_t \xx_j + \textbf{X}_t \bb_j \}} \right]^{I_{tk}}
\end{align*}
where $I_{tk} = \mathbbm{1}(z_t = k)$.
Denote $\boldsymbol{W}_t = (\boldsymbol{Z}_t, \boldsymbol{X}_t)$ and $\zz_j = (\xx_j, \bb_j)$, so that
\begin{equation*}
\label{eq:cond11}
    p(\zz \mid  -) = \prod_{k=1}^{K} p(\zz_k) \times \prod_{t=1}^{T} \prod_{k=1}^K \left[ \dfrac{exp \{ \boldsymbol{W}_t \zz_k \}}{1+ \sum_{j=1}^{K-1}exp \{  \boldsymbol{W}_t \zz_j \}} \right]^{I_{tk}}.
\end{equation*} 
Using the results on the conditional likelihood shown by \cite{held2006bayesian},
\begin{align*}
    &p(\zz_j \mid \zz_{-j}, -)  \propto p(\zz_j) \times \\
    &\times \prod_{t=1}^{T} \left[ \dfrac{exp \{  \boldsymbol{W}_t \zz_j - C_{tj} \}}{1+ exp \{ \boldsymbol{W}_t \zz_j - C_{tj}\}} \right]^{I_{tj}} \left[ \dfrac{1}{1+exp \{ \boldsymbol{W}_t \zz_j  - C_{tj}\}} \right]^{1-I_{tj}}
\end{align*}
\begin{equation}
    = p(\zz_j) \times \prod_{t=1}^{T}  \dfrac{[exp \{  \boldsymbol{W}_t \zz_j - C_{tj} \}]^{I_{tj}}}{[1+ exp \{ \boldsymbol{W}_t \zz_j - C_{tj}\}]} \label{eq:cond2}.
\end{equation}
for $C_{tj} = log \sum_{k \neq j} exp(\boldsymbol{W}_t \zz_k)$.
We can use the fact that, as in \citep{polson2013bayesian}, for $\omega$ distributed according to a Pólya-Gamma with parameters $b > 0$ and $0$ (i.e. $\omega \sim PG(b,0)$), $\eta$ a linear function of predictors and with $\kappa = a - b/2$, 
\begin{equation*}
\label{s:pg}
    \dfrac{exp\{ \eta \}^a}{(1 + exp\{ \eta \})^b} = 2^{-b} exp\{ \kappa \eta \} \int_0^{\infty} exp \{- \omega \eta^2 / 2 \} p(\omega) d \omega,
\end{equation*}
  to write (\ref{eq:cond2}) as
\begin{align*}
   & \propto p(\zz_j) \times
\prod_{t=1}^{T} exp\{ \kappa_{tj} (\boldsymbol{W}_t \zz_j - C_{tj})  \} \times \\
   &\times \int_0^{\infty} exp \{- \omega_{tj} (\boldsymbol{W}_t \zz_j - C_{tj})^2 / 2\} p(\omega_{tj}) d \omega_{tj}.
\end{align*}
using $a = I_{tj}$ and $b=1$.
Finally, by Theorem 1 from \cite{polson2013bayesian} we have that, for  a Gaussian prior $\zz_j \sim N(m_0, I \cdot \delta^2_0)$ and denoting by $\boldsymbol{\Omega}_{j}$ the diagonal matrix with the $\omega_{tj}$ in the diagonal,
 \begin{align*}
 &\omega_{t j} \mid \zz_j \sim \mathrm{PG}\left(1, \boldsymbol{W}_t \zz_j - C_{tj} \right) \\
     &\zz_{j} \mid \boldsymbol{\Omega}_{j} \sim N\left(m_{j}, V_{j}\right),
 \end{align*}
with $$V_{j}^{-1}=\left(\boldsymbol{W}_j^{\prime} \boldsymbol{\Omega}_{j} \boldsymbol{W}_j+ (I \cdot \sigma_{0}^{2})^{-1} \right)$$
and 
$$m_{j}=V_{j}\left(\boldsymbol{W}_j^{\prime}\left(\left(\boldsymbol{I}_j-1 / 2\right)-\right.\right.\left.\left.\boldsymbol{\Omega}_{j} \boldsymbol{C}_{j}\right)+(\sigma_0^2)^{-1} a_{j}\right)$$.

\newpage
\section*{Web Appendix C: Posterior full conditionals and Algorithm}
Detailed Gibbs sampler steps are as follows. 
The complete MCMC is summarized in Algorithm 1.

$\boldsymbol{p(\boldsymbol{\zeta}_j \mid -)}$
\\ Following the reasoning that appears in the Appendix, 
 \begin{align}
     &\zz_{j} \mid \boldsymbol{\Omega}_{j} \sim N\left(m_{j}, V_{j}\right) \label{eq:postzeta} \\
     &\omega_{t j} \mid \zz_j \sim \mathrm{PG}\left(1, \boldsymbol{W}_t \zz_j - C_{tj} \right), \nonumber
 \end{align}
with $$V_{j}^{-1}=\left(\boldsymbol{W}_j^{\prime} \boldsymbol{\Omega}_{j} \boldsymbol{W}_j+ (I \cdot \sigma_{0}^{2})^{-1} \right)$$
and 
$$m_{j}=V_{j}\left(\boldsymbol{W}_j^{\prime}\left(\left(\boldsymbol{I}_j-1 / 2\right)-\right.\right.\left.\left.\boldsymbol{\Omega}_{j} \boldsymbol{C}_{j}\right)+(\sigma_0^2)^{-1} m_{j}\right).$$

$\boldsymbol{p(\psi_j \mid -)}$
\\ Define $\psi_j = (\mu_j, \Sigma_j)$. 
We use the common normal-inverse Wishart model so that the update consists of drawing the variance covariance matrix from an inverse Wishart and the mean vector from a multivariate Gaussian conditioned on the drawn variance-covariance matrix.
\begin{align}
\Sigma_j \mid \zbt, \ybt & \sim \text { Inv-Wishart }_{\nu_{n_j}}\left(\Lambda_{n_j}\right) \label{eq:postsigma} \\
\mu_j \mid \Sigma_j, \zbt, \ybt & \sim \mathrm{N}\left(\mu_{n_j}, \Sigma_{j} \right) \label{eq:postmu}
\end{align}
with $\nu_{n_j}$, $\Lambda_{n_j}$ and $\mu_{n_j}$ as in \cite{gelman1995bayesian}.

$\boldsymbol{p(\ybm \mid -)}$
\begin{equation}
\label{eq:postym}
\textbf{y}_{t} \mid z_t = j \sim N(\mu_j, \Sigma_j) \qquad \forall t \text{ s.t. } m_t = 1
\end{equation}

$\boldsymbol{p(z_t \mid -)}$
\begin{equation}
\label{eq:postz}
z_t \mid Q, \mu, \Sigma , z_{t-1}, z_{t+1}, y_t \sim \operatorname{Multi}\left(\frac{q_{z_{t-1}, 1} f_1 q_{1, z_{t+1}} }{\sum_{j=1}^K q_{z_{t-1}, j} f_j q_{j, z_{t+1}} }, \ldots, \frac{q_{z_{t-1}, K} f_K q_{K, z_{t+1}} }{\sum_{j=1}^K q_{z_{t-1}, j} f_j q_{j, z_{t+1}} }\right),
\end{equation}
where $f_j = N(y_t \mid \mu_j, \Sigma_j)$.
\begin{algorithm}
\caption{Gibbs sampler for NHMM for PA data}
\begin{algorithmic}
\label{alg:all}
\Require A $T \times d$ matrix $Y$ of observed outcomes and entries for missing observations, a $T \times p$ design matrix $X$, fixed $K$, $m$, $\#init$, $\#burn$, $\#iters$.
\State Initialize $\boldsymbol{\zeta}^{(0)}, \mu_1^{(0)}, ..., \mu_K^{(0)}, \Sigma_1^{(0)}, ..., \Sigma_K^{(0)}, z_0^{(0)}, ..., z_T^{(0)}$.
\For{$i$ from $1$ to $\# init$}
\For{$j$ from 1 to $K$}
    \State Sample $\Sigma_{j}^{(i)}$ and $\boldsymbol{\mu_{j}}^{(i)}$ as in \ref{eq:postsigma} and \ref{eq:postmu},
\EndFor
    \State Sample $\ybm^{(i)}$ as in \refeq{eq:postym}, $\boldsymbol{\zeta}^{(i)}$ as in \refeq{eq:postzeta}, $\zbt^{(i)}$ as in \refeq{eq:postz},
\EndFor
\For{$j$ from 1 to $K$}
\State Set $Y_{aug,j} = (\bar{\mu}_{j}, ..., \bar{\mu}_{j})$ and $z_{aug,j}=(j,...,j)$ vectors of length 96 (or alternative length of daily data depending on the chosen data structure) $\times m$
\EndFor
\State Set $Y_{aug} = (Y, Y_{aug, 1}, ..., Y_{aug, K})$ and $z_{aug} = (\zbt, z_{aug, 1}, ..., z_{aug, K})$
\For{$i$ from $(\# init + \#burn + 1)$ to $(\# init + \#burn + \#iters)$}
\For{$j$ from 1 to $K$}
    \State Sample $\Sigma_{j}^{(i)}$ and $\mu_{j}^{(i)}$ as in \ref{eq:postsigma} and \ref{eq:postmu} using $Y_{aug}$ and $z_{aug}$,
\EndFor
    \State Sample $\ybm^{(i)}$ as in \ref{eq:postym}, $\boldsymbol{\zeta}^{(i)}$ as in \ref{eq:postzeta} using $z_{aug}$,
     $\zbt^{(i)}$ as in \ref{eq:postz}
    \If{$i > (\#init + \#burn)$}
    \State Save $\boldsymbol{\Sigma}^{(i)}$, $\boldsymbol{\mu}^{(i)}$, $\ybm^{(i)}$, $\boldsymbol{\zeta}^{(i)}$, $\zbt^{(i)}$
    \EndIf
\EndFor
\State \Return saved $\boldsymbol{\Sigma}^{(i)}$, $\boldsymbol{\mu}^{(i)}$, $\ybm^{(i)}$, $\boldsymbol{\zeta}^{(i)}$, $\zbt^{(i)}$.
\end{algorithmic}\label{algo}
\end{algorithm}


\section*{Web Appendix D: Data augmentation prior}
The data augmentation prior aims at stabilizing the learning process of the subset of $\boldsymbol{\zeta}$ parameters associated to hours presenting some degree of separation or lack of representation of certain states, this way ensuring a more stable learning process for the entire set of $\zz$.
Such prior is specified by augmenting the observed data with synthetic pseudo-data representing $m \times K$ days of specified values of $(y_t, z_t)$.
Values of $(y_t, z_t)$ are chosen to ensure representation of every state at every hour while also not imposing any particular activity pattern. 
More precisely, a set of assumed-to-be-observed $z_t$ having one $z_t = j$ for each $j$ in $1, ..., K$ and for each $t = 1, ..., 24 \times 4$ (where $24 \times 4$ is the number of daily hour quarters) is created.
For what concerns values of $y_t$, the complete MCMC algorithm is first run for 2000 initial iterations using the observed data only to get a first estimate of states emission parameters $\boldsymbol{\psi}_j$ for $j = 1, ...,K$. 
The obtained initial estimates for $\boldsymbol{\psi}_j$ are then used to generate values of $y_t$ coherently with the selected $z_t$.
This synthetic dataset is then replicated $m$ times, where $m$ controls the prior's strength, with higher $m$ encoding stronger prior information that mitigates the hour of day effect. 
When this initialization process is completed, the MCMC is initiated again using the augmented dataset.


\newpage
\section*{Web Appendix E: Complement to Hybrid Dataset Results}
Using an hybrid dataset as explained in Section 6 of the main manuscript does not allow to fully control for the strength of the connection $x \to \zbt$, since this relationship has to be taken as it exists in the data. 
Still, we construct a coefficient to capture such strength as $c \cdot \prod_{k=1}^{K} (max(p(z = k) - min(p(z=k)))$ where $c$ is the number of times within the daily 24 hours in which the most represented state (the one with the higher percentage) has changed according to the posterior estimates of daily behaviors (i.e. the coefficient is evaluated on the posterior marginal distributions for each patient). 
In Figure \ref{fig:combined2}, patients IDs are ordered according to such coefficient on the $x$-axis, paired with an increasing value of $\gamma$ used to construct the hybrid dataset so that the $x$-axis shows simultaneously increasing strength of $x \to \zbt$ and $x \to \mbt$ in the way described in the main manuscript.
\begin{figure}[htbp]
  \centering
  \begin{minipage}{0.45\textwidth} 
    \caption*{Medium \% of missingness} 
    \includegraphics[width=\linewidth]{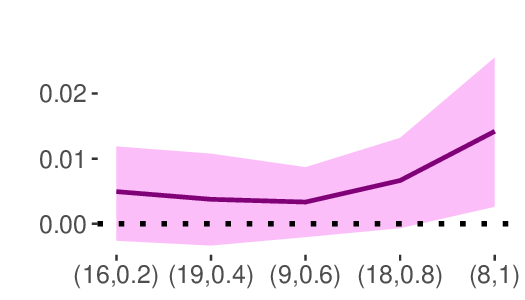}
  \end{minipage}
  \hspace{1cm} 
  \begin{minipage}{0.45\textwidth}
    \caption*{High \% of missingness} 
    \includegraphics[width=\linewidth]{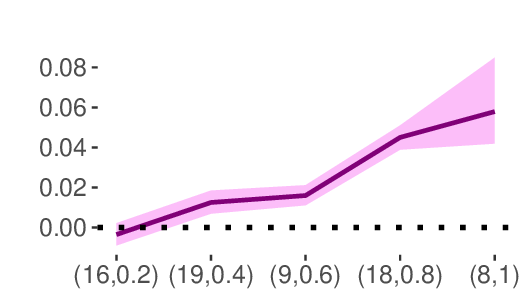}
  \end{minipage}
  \caption{Comparison of missing data imputation performance. Higher values correspond to a better performance by the NHMM relative to the HMM. The $x$ axis corresponds too (ID, $\gamma$) pairs.}
  \label{fig:combined2}
\end{figure}
While the trend with the medium missing percentage is less distinct, the case of high missing percentage clearly shows how the simultaneous increasing role played by the hour of the day on the availability of the data and on the states makes the NHMM more efficient of the HMM for missing data imputation.

Figure \ref{fig:mm} instead visualizes the estimated marginal probabilities of being in a certain state over the course of the 24 hours in a day for the two individuals which the figures in the main manuscript are referred to (ID8 and ID16). 
The goal is to show evidence that the two individuals for which the results are shown are not two individuals with exceptionally high strength of the $x \to \zbt$, this ruling out the hypothesis of the shown results being due to this feature.
ID8 shows quite some differences over the course of the day in terms of what is the most prevalent state, ID16 which shows instead stability in those terms. 
Hence, the two individuals can be considered as examples with clear different strength of the $x \to \zbt$ connection and the results cannot be interpreted as being due to some exceptionality in their PA routine. 
Similar results were obtained with other individuals that were part of the constructed hybrid dataset.
\begin{figure}[htbp]
  \centering
  \begin{minipage}{0.45\textwidth} 
    \caption*{ID8} 
    \includegraphics[width=\linewidth]{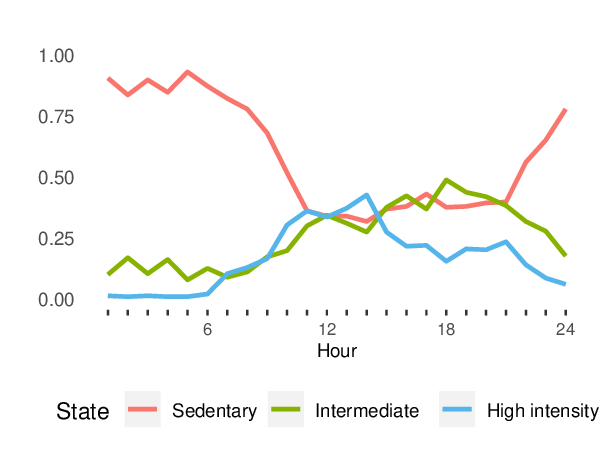}
  \end{minipage}
  \hspace{1cm} 
  \begin{minipage}{0.45\textwidth}
    \caption*{ID16} 
    \includegraphics[width=\linewidth]{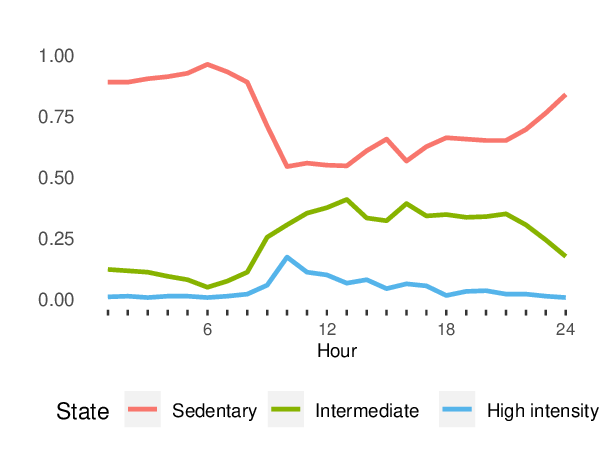}
  \end{minipage}
  \caption{Marginals probabilities of being in each state per hour. This estimate of the NHMM corresponds to the "medium percentage missingness" case of the hybrid dataset.}
  \label{fig:mm}
\end{figure}

\newpage

\section*{Web Appendix F: Simulations with fully generated data}
\label{s:sim}
We assume $\yb$ to be a bivariate vector of observations mimicking observations of steps and heart rate as collected by Fitbit and aggregated in time blocks $t = 1, ..., T$ of 15 minutes each.
We set $x_t$ to be the hour of the day corresponding to $t$. 
We simulate data according to transition matrices whose vectors of entries are obtained as 
\begin{equation}
   \boldsymbol{q}_t = \boldsymbol{q}_{0,t} + \nu \cdot \boldsymbol{q}_{h,t},
\end{equation}
where $\boldsymbol{q}_0$ is a vector of entries such that the generated data correspond to the case $x \perp \zbt$, while $\boldsymbol{q}_h$ is a vector of entries corresponding to the case in which the hour of the day is highly relevant for determining the underlying state.  
The parameter $\nu$ is set to govern the degree of influence of hour of the day in determining transitions, i.e. the strength of $x \to \zbt$. 
After generating full time series of $\yb$, we simulate missing data through artificially deleting simulated values according to a $Ber(\boldsymbol{p}_t)$ distribution with 
\begin{equation}
\boldsymbol{p}_t = (1-\gamma) \cdot \boldsymbol{p}_{0,t} + \gamma \cdot \boldsymbol{p}_{h,t},
\end{equation}
with $\boldsymbol{p}_0$ corresponding to the case of missingness being MCAR and $\boldsymbol{p}_h$ corresponding to the case of $x \to \mbt$ being a very strong connection. Intuitively, the higher $\gamma$ and $\nu$, the higher the relevance of the hour of the day in determining $\zbt$ and $\mbt$. 

Those simulations are meant to test the practical implications of our theoretical claims.
In particular, we want to test to what extent and under what circumstances the NHMM allows for tangibly better performances when compared to a classic HMM.
To perform such evaluation, we compare the HMM to our NHMM both in terms of parameter estimation and missing data imputation. 
With fully simulated data, we have much more room for comparing the outcomes to the true values than when using the hybrid dataset and to directly control the strength of $x \to \zbt$.
The former is measured according to the following metrics; the average Frobenius distance between the true and estimated transition matrices, the average distance between the estimated and true marginal probability of membership in each state, and the percentage of $\zbt$ that have been allocated to the correct true latent state. 
Performance of missing data imputation is measured looking at the RMSE of imputed vs true values (for both heart rate and the number of steps) and the percentage of missing observations that have their latent state $z_t$ allocated correctly. 

Results are displayed in Figure 
\ref{fig:combined3}. 
In those visualizations, increasing values on the $x$-axis correspond to an increasing strength in $x \to \zbt$ and $x \to \mbt$, while values on the $y$-axis correspond to the comparison in terms of performance between the two models, with higher values on the latter corresponding to the NHMM outperforming the HMM for all the metrics considered.
The two models perform similarly in correspondence of the dotted horizontal line.

The results confirm our theoretical claims, with the NHMM performing better whenever $x$ becomes relevant for $\mbt$ and $\zbt$, this being true for all our metrics and for both parameters estimation and missing data imputation. 
The key outcome is again that, the more the data generating process resembles a scenario not breaking ignorability, the NHMM is not necessarily a needed choice, with it becoming instead needed when the data generating process breaks condition 3(b) from the main manuscript. 
When both the strength of $x \to \mbt$ and $x \to \zbt$ is low or negligible, a model with less parameters such as the HMM is even preferable. 
The higher complexity and number of parameters of the NHMM has to be preferred when $\gamma$ and $\nu$ are higher, this being the case of PA data obtained in free living conditions as elaborated in Section 2 of the main paper.
Note that tests were also performed under scenarios of mixed strength of $\gamma$ and $\nu$, bringing to comparable results.

\begin{figure}[htbp]
  \centering
  \begin{minipage}{0.33\textwidth} \subcaption{Steps} 
    \includegraphics[width=\linewidth]{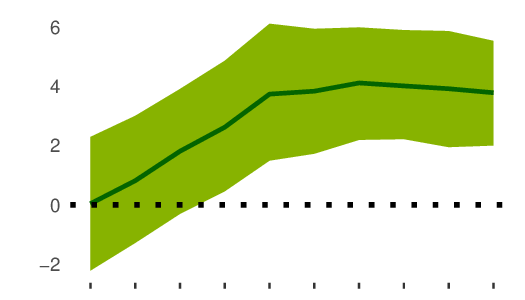}
  \end{minipage}%
  \begin{minipage}{0.33\textwidth}
    \subcaption{Heart rate} 
    \includegraphics[width=\linewidth]{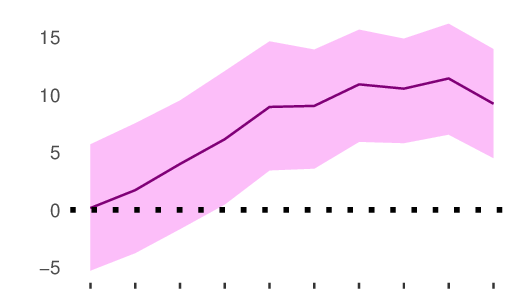}
  \end{minipage}%
  \begin{minipage}{0.33\textwidth}
    \subcaption{State allocation (all)}
    \includegraphics[width=\linewidth]{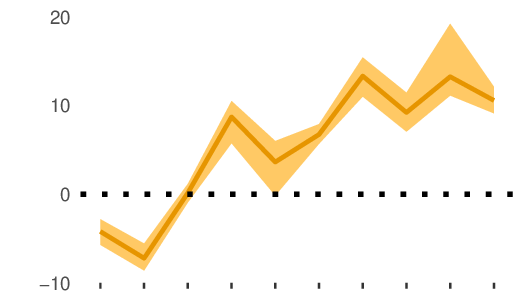}
  \end{minipage}

  \vspace{0.5cm} 
  \begin{minipage}{0.33\textwidth}
    \subcaption{Transition matrices}
    \includegraphics[width=\linewidth]{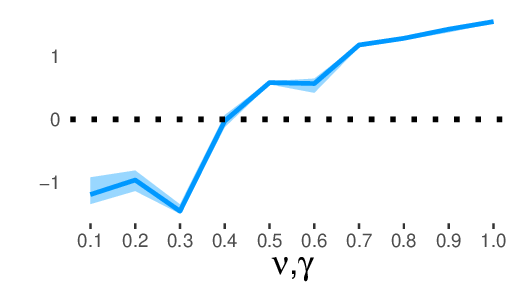}
  \end{minipage}%
  \begin{minipage}{0.33\textwidth}
    \subcaption{Marginal per hour}
    \includegraphics[width=\linewidth]{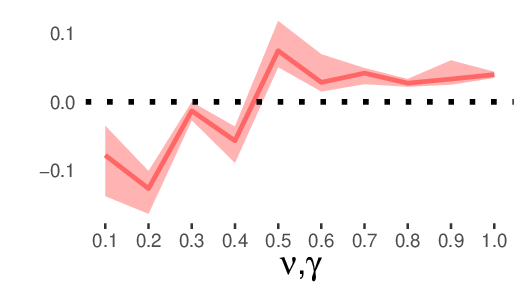}
  \end{minipage}%
  \begin{minipage}{0.33\textwidth}
    \subcaption{State allocation ($\mbt$)}
\includegraphics[width=\linewidth]{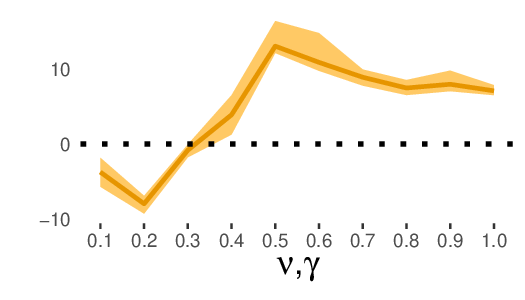}
  \end{minipage}

  \caption{Comparing HMM and NHMM under different scenarios captured by $\nu$ and $\gamma$, using different performance metrics. The y-axis represents the difference in performance between the two models; the higher the value on the y, the higher the performance of the NHMM relatively to the HMM. The dotted line corresponds to the performance of the two models being equivalent.}
  \label{fig:combined3}
\end{figure}

\newpage
\section*{Web Appendix G: Convergence}
Here we show the trace-plots showing the convergence of the parameters $\boldsymbol{\zeta}$ corresponding to the inference shown in Figure 4 of the main manuscript.
The first two figures correspond to the traceplots obtained when our prior on the marginal state probabilities was used.
The second two plots correspond to the case in which such strategy was not in place.
\begin{figure}[htbp]
    \centering
\includegraphics[width=\textwidth]{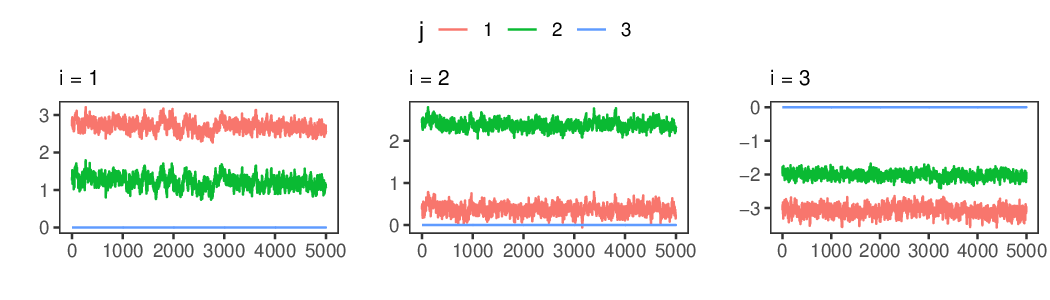}
\caption{Posterior draws for $\boldsymbol{\xi_{ij}}$ with data augmentation prior.}
    \label{fig:traces}
\end{figure}

\begin{figure}[htbp]
    \centering
\includegraphics[width=\textwidth]{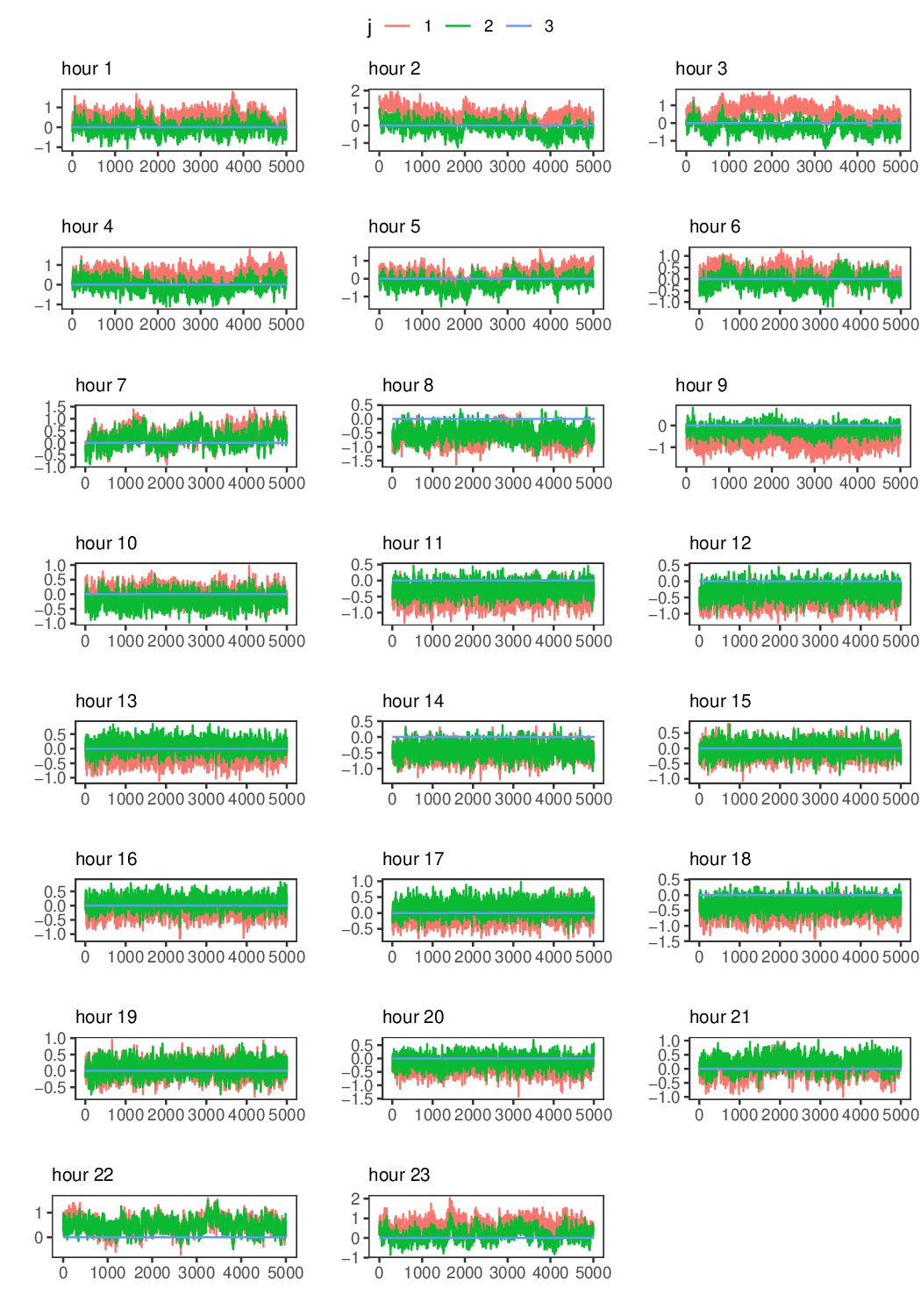}
\caption{Posterior draws for $\boldsymbol{\beta_j}$ with data augmentation prior.}
    \label{fig:traces22}
\end{figure}
\begin{figure}[htbp]
    \centering
\includegraphics[width=\textwidth]{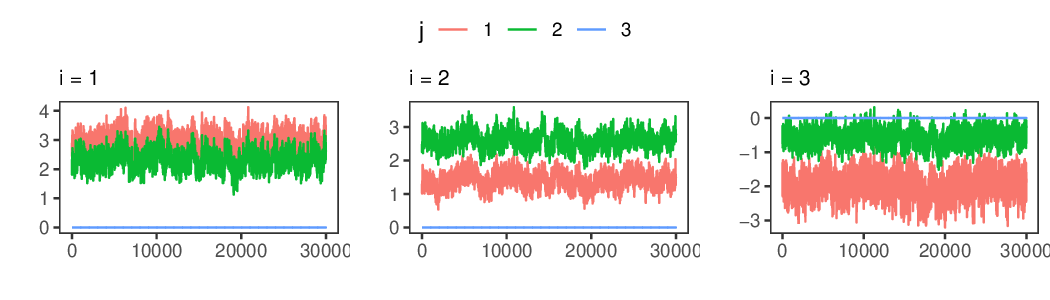}
\caption{Posterior draws for $\boldsymbol{\xi_{ij}}$ without data augmentation prior.}
    \label{fig:traces}
\end{figure}

\begin{figure}[htbp]
    \centering
\includegraphics[width=\textwidth]{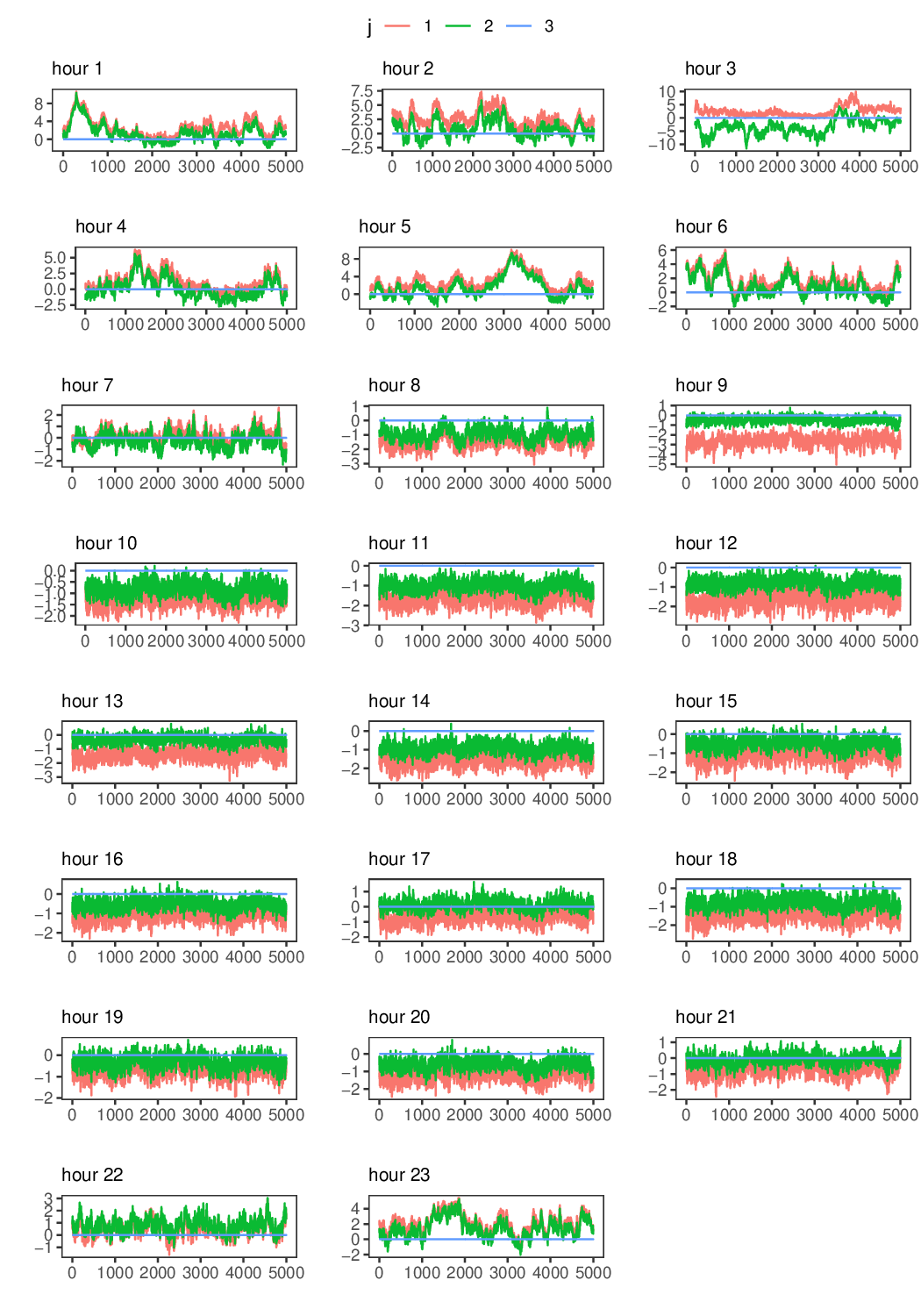}
\caption{Posterior draws for $\boldsymbol{\beta_j}$ without data augmentation prior.}
    \label{fig:traces22}
\end{figure}